\documentclass[preprint,12pt]{article}
\usepackage{graphicx,cite}
\usepackage{multicol,multirow}
\usepackage{amsmath,amssymb,amsthm}
\usepackage{graphicx}
\usepackage{ifpdf}
\usepackage{color}
\allowdisplaybreaks
\DeclareGraphicsExtensions{.pdf,.png,.jpg,.jpeg,.mps,.eps}

\usepackage[colorlinks=false,linktocpage=false]{hyperref}
\def\pagenumberfont#1{\def\pn@font{#1}}      \def\pn@font{}

\catcode`\@=11
\def\doublearrow{\@ifnextchar [{\@doublearrow }{\@doublearrow[0]}}
\def\@doublearrow[#1]{\mathrel{\,\lower0.15ex
		\hbox{\let\@linefnt\eightln\unitlength0.6ex\begin{picture}(4,3)
				\ifcase#1\put(4,0.8){\vector(-1,0){4}}\put(0,2){\vector(1,0){4}}
				\or      \put(0,0.8){\vector(1,0){4}}\put(0,2){\vector(1,0){4}}
				\or      \put(4,0.4){\vector(-2,1){4}}\put(0,0.4){\vector(2,1){4}}
				\or      \put(0,2.4){\vector(2,-1){4}}\put(0,0.4){\vector(2,1){4}}
				\fi
		\end{picture}}\,}}
\catcode`\@=12

\font\eightln=line10 at8pt \catcode`\@=11
\def\triplearrow{\@ifnextchar [{\@triplearrow }{\@triplearrow[0]}}
\def\@triplearrow[#1]{\mathrel{\,\lower0.15ex
		\hbox{\let\@linefnt\eightln\unitlength0.6ex\begin{picture}(4,3)
				\ifcase#1\put(0,0.3){\vector(1,0){4}}\put(0,1.5){\vector(1,0){4}}
				\put(0,2.7){\vector(1,0){4}}
				\or\put(0,0){\vector(2,1){4}}\put(0,1){\vector(2,1){4}}
				\put(0,3){\vector(4,-3){4}}
				\or\put(0,3){\vector(2,-1){4}}\put(0,2){\vector(2,-1){4}}
				\put(0,0){\vector(4,3){4}}
				\or\put(0,0){\vector(4,3){4}}\put(0,1.5){\vector(1,0){4.5}}
				\put(0,3){\vector(4,-3){4}}
				\or\put(0,0.3){\vector(1,0){4}}\put(0,1){\vector(2,1){4}}
				\put(0,3){\vector(2,-1){4}}
				\or\put(0,0){\vector(2,1){4}}\put(0,2){\vector(2,-1){4}}
				\put(0,2.7){\vector(1,0){4}}
				\or\put(4,0.3){\vector(-1,0){4}}\put(4,1.5){\vector(-1,0){4}}
				\put(0,2.7){\vector(1,0){4}}
				\or\put(4,2){\vector(-2,-1){4}}\put(4,3){\vector(-2,-1){4}}
				\put(0,3){\vector(4,-3){4}}
				\or\put(4,0){\vector(-1,0){4}}\put(4,3){\vector(-2,-1){4}}
				\put(0,3){\vector(2,-1){4}}
				\or\put(4,0.3){\vector(-1,0){4}}\put(0,1.5){\vector(1,0){4}}
				\put(0,2.7){\vector(1,0){4}}
				\or\put(0,0.3){\vector(1,0){4}}\put(4,1.5){\vector(-1,0){4}}
				\put(0,2.7){\vector(1,0){4}}\fi
		\end{picture}}\,}}
\catcode`\@=12

\font\eightln=line10 at8pt
\catcode`\@=11
\def\tripline{\@ifnextchar [{\@tripline }{\@tripline[0]}}
\def\@tripline[#1]{\mathrel{\,\lower0.15ex
		\hbox{\let\@linefnt\eightln\unitlength0.6ex\begin{picture}(4,3)
				\ifcase#1\put(0,0.3){\line(1,0){4}}\put(0,1.5){\line(1,0){4}}
				\put(0,2.7){\line(1,0){4}}
				\or\put(0,0){\line(2,1){4}}\put(0,1){\line(2,1){4}}
				\put(0,3){\line(4,-3){4}}
				\or\put(0,3){\line(2,-1){4}}\put(0,2){\line(2,-1){4}}
				\put(0,0){\line(4,3){4}}
				\or\put(0,0){\line(4,3){4}}\put(0,1.5){\line(1,0){4.5}}
				\put(0,3){\line(4,-3){4}}
				\or\put(0,0.3){\line(1,0){4}}\put(0,1){\line(2,1){4}}
				\put(0,3){\line(2,-1){4}}
				\or\put(0,0){\line(2,1){4}}\put(0,2){\line(2,-1){4}}
				\put(0,2.7){\line(1,0){4}}
				\fi
		\end{picture}}\,}}
\catcode`\@=12

\newcommand{\asr}{\doublearrow}
\newcommand{\psr}{\doublearrow[1]}

\makeatletter
\newcommand{\contraction}[5][1ex]{%
	\mathchoice
	{\contraction@\displaystyle{#2}{#3}{#4}{#5}{#1}}%
	{\contraction@\textstyle{#2}{#3}{#4}{#5}{#1}}%
	{\contraction@\scriptstyle{#2}{#3}{#4}{#5}{#1}}%
	{\contraction@\scriptscriptstyle{#2}{#3}{#4}{#5}{#1}}}%
\newcommand{\contraction@}[6]{%
	\setbox0=\hbox{$#1#2$}%
	\setbox2=\hbox{$#1#3$}%
	\setbox4=\hbox{$#1#4$}%
	\setbox6=\hbox{$#1#5$}%
	\dimen0=\wd2%
	\advance\dimen0 by \wd6%
	\divide\dimen0 by 2%
	\advance\dimen0 by \wd4%
	\vbox{%
		\hbox to 0pt{%
			\kern \wd0%
			\kern 0.5\wd2%
			\contraction@@{\dimen0}{#6}%
			\hss}%
		\vskip 0.2ex%
		\vskip\ht2}}

\newcommand{\contraction@@}[3][0.06em]{%
	\hbox{%
		\vrule width #1 height 0pt depth #3%
		\vrule width #2 height 0pt depth #1%
		\vrule width #1 height 0pt depth #3%
		\relax}}

\newlength{\myheight}
\newlength{\mydepth}

\makeatletter
\newcommand{\lcontraction}[5][1ex]{%
	\mathchoice
	{\lcontraction@\displaystyle{#2}{#3}{#4}{#5}{#1}}%
	{\lcontraction@\textstyle{#2}{#3}{#4}{#5}{#1}}%
	{\lcontraction@\scriptstyle{#2}{#3}{#4}{#5}{#1}}%
	{\lcontraction@\scriptscriptstyle{#2}{#3}{#4}{#5}{#1}}}%
\newcommand{\lcontraction@}[6]{%
	\setbox0=\hbox{$#1#2$}%
	\setbox2=\hbox{$#1#3$}%
	\setbox4=\hbox{$#1#4$}%
	\setbox6=\hbox{$#1#5$}%
	\dimen0=\wd2%
	\advance\dimen0 by \wd6%
	\divide\dimen0 by 2%
	\advance\dimen0 by \wd4%
	\setlength{\myheight}{-#6}
	\addtolength{\myheight}{-0.2ex}
	\addtolength{\myheight}{0.06em}
	\setlength{\mydepth}{#6}
	\addtolength{\mydepth}{0.2ex}
	\vbox{%
		\vskip 0.2ex%
		\vskip\ht2
		\hbox to 0pt{%
			\kern \wd0%
			\kern 0.5\wd2%
			\lcontraction@@{\dimen0}{\myheight}{\mydepth}%
			\hss
		}%
}}

\newcommand{\lcontraction@@}[4][0.06em]{%
	\hbox{%
		\vrule width #1 height -0.2ex depth #4%
		\vrule width #2 height #3 depth #4%
		\vrule width #1 height -0.2ex depth #4%
		\relax}}
\makeatother

\newtheorem{remark}{Remark}[section]
\newtheorem{theorem}{Theorem}[section]
\newtheorem{corollary}{Corollary}[section]
\newtheorem{lemma}{Lemma}[section]                                                                                                              

\newtheorem{definition}{Definition}[section]

\numberwithin{equation}{section}

\def\R{\mathbb R}
\def\C{\mathbb C}

\def\Z{\mathbb Z}
\def\G{\mathsf{G}}
\def\H{\mathbb H}

\def\X{\mathcal X}

\def\P{{\mathcal P}}

\newcommand{\eps}{\varepsilon}
\def\tr{{\rm tr}}
\newcommand{\PSL}{{\rm PSL}}
\newcommand{\SL}{{\rm SL}}

\def\T{{\mathcal T}}

\allowdisplaybreaks

\begin{document}

\title{Spectral form factor in the Hadamard-Gutzwiller model: orbit pairs contributing in the third order}
\author{{\sc Huynh M. Hien} \\[1ex] 
	Department of Mathematics and Statistics,\\
 Quy Nhon University, \\ 	170 An Duong Vuong,
 Quy Nhon, Vietnam\\
	e-mail: huynhminhhien@qnu.edu.vn
	\date{} 
}
\maketitle

\begin{abstract} In this paper we consider orbit pairs contributing in the third order of the spectral form factor in the Hadamard-Gutzwiller model.  We prove that periodic orbits including two 2-encounters in certain structures have partner orbits. The action differences are estimated at $\ln(1+u_1s_1)(1+u_2s_2)$ with
	explicit error bounds, where $(u_1,s_1)$ and $(u_2,s_2)$ are the coordinates of the piercing points. A new symbolic dynamics for  orbit pairs via conjugacy classes is also provided. 
	
\end{abstract}

\section{Introduction}

	In quantum chaos, there is considerable interest in understanding statistics associated to periodic orbits
	since these are related to eigenvalue statistics through trace formulae. Special attention has been given to the {\em spectral form factor}, which is expressed by a double sum over periodic orbits
	\begin{equation}\label{formfactor}
		K(\tau)=\Big\langle \frac{1}{T_H}\sum_{\gamma,\gamma'}A_\gamma 
		A_{\gamma'}^* e^{\frac{i}{\hbar}(S_\gamma-S_{\gamma'})}
		\delta\Big(\tau T_H-\frac{T_\gamma+T_{\gamma'}}{2}\Big) \Big\rangle,
	\end{equation}
	where $\langle\cdot\rangle$ abbreviates the  average over the energy and
	over a small time window, $T_H$ denotes the Heisenberg time and $A_\gamma$, $S_\gamma$, and $T_\gamma$ 
	are the amplitude, the action, and the period of the orbit $\gamma$, respectively.

	The diagonal approximation  $\gamma=\gamma'$ to (\ref{formfactor}) studied by Hannay/Ozorio de Almeida \cite{hoda} 
	and Berry \cite{berry} in the 1980's contributes to the first order term $2\tau$; see also \cite{KeatRob}. The efforts of researchers have been to understand higher order effects. 
	To the next orders, as $\hbar\to 0$, the main term  
	from (\ref{formfactor}) arises owing to those orbit pairs $\gamma\neq\gamma'$ 
	for which the action difference $S_\gamma-S_{\gamma'}$ is `small'. 	In 2001, an influential heuristic work of Sieber and Richter \cite{SieberRichter}
	who predicted that a given periodic orbit with a small-angle self-crossing in configuration space 
	will admit a partner orbit with almost the same action.	
	 The original orbit and its partner are then called  a Sieber-Richter pair. In phase space, a Sieber-Richter pair contains a region where two stretches of each  orbit are almost mutually time-reversed
	and one addresses this region as a {\em $2$-encounter} or, more strictly, a {\em $2$-antiparallel encounter};
	the `2' stands for two orbit stretches which are close in configuration space, 
	and `antiparallel' means that the two stretches have opposite directions. It was shown in \cite{SieberRichter} that 
	Sieber-Richter pairs contribute to the spectral form factor \eqref{formfactor} the second order term
	$-2\tau^2$, and it turned out that the result agreed 
	with what is obtained using random matrix theory \cite{efetov}, for certain symmetry classes. 
	The work by Sieber and Richter has led to the important and difficult problem of understanding this phenomenon is more detail and more rigorously in particular classes of systems. Until 2012, Gutkin and Osipov \cite{GO} analysed Sieber-Richter pairs for the Baker map, which admits very transparent symbolic dynamics, in a combinatorial way. 
	
	Most contribution in this subject matter is M\"uller et al. In a series of works \cite{HMBH,mueller2009,mueller2004,mueller2005}, the authors provided an expansion 
	to all orders in $\tau$
	\[ K(\tau)=2\tau-\tau\,\ln(1+2\tau)=2\tau-2\tau^2+2\tau^3+\ldots \] 
	for the symmetry class relevant for time-reversal invariant systems, by including the higher-order encounters also; see also \cite{haake,muellerthesis}.  It was shown in \cite{HMBH} that there are five families of pairs of orbits responsible for the third order $\tau^3$, namely three families of
	orbit pairs differing in two 2-encounters and two families of orbit pairs differing in one single 3-encounter. 
 Periodic orbits with encounters have partners obtained by reconnections stretches inside encounter area owing to the hyperbolicity. However, the existence of partner orbits and estimates of the action differences are still missing.

	To establish a more detailed mathematical understanding, it is necessary to consider the classical side and try to prove the existence of partner orbits and derive good estimates for the action differences 
	of the orbit pairs. For $2$-antiparallel encounters this was done in \cite{HK,Huynh17}, where the authors considered 
	the geodesic flow on compact factors of the hyperbolic plane; in this case the action of a periodic orbit is half of its length/period. It was shown in \cite{HK} that
	a $T$-periodic orbit of the geodesic flow crossing itself in 
	configuration space at a time $T_1$ has a unique partner orbit that remains 
	$9|\sin(\phi/2)|$-close to the original one and
	the action difference between them is approximately equal
	$\ln(1-(1+e^{-T_1})(1+e^{-(T-T_1)})\sin^2(\phi/2)))$ with the error bound
	$12\sin^2(\phi/2)e^{-T}$, where $\phi$ is the crossing angle,
	and  this proved the accuracy of Sieber/Richter's prediction in \cite{SieberRichter} mentioned above.
	For higher-order encounters, Huynh \cite{Huynh16} shows that  there exist $(L-1)!-1$ partner orbits for a given periodic orbit with an $L$-parallel encounter such that any two piercing points are not too close
	and provided estimates for the action differences.

	In the present paper we continue considering the geodesic flow on compact factor of the hyperbolic plane,
	which is a compact Riemann surface of constant curvature of genus at least two.
	In the physics community this system is often called the Hadamard-Gutzwiller model, 
	and it has frequently been studied \cite{braun2002,HMBH,Sieber1}; 
	further related work includes \cite{haake,mueller2005,Turek05}. We prove
	the existence of the partner orbit which differs in both encounters for a given periodic orbit including two 2-encounters with piercing points having  coordinates $(u_1,s_1), (u_1,u_2)$ in certain distributions. The action differences of orbit pairs of all cases are estimated at $\ln(1+u_1s_1)(1+u_2s_2)$ with  explicit error bounds. This paper also provides a new symbolic dynamics for orbit pairs via conjugacy classes.

The paper is organized as follows. In Section 2 we recall background and materials, including Poincar\'e sections, the Anosov and closing lemmas, conjugacy classes and rigorous definitions of encounters, partners in the Hadamard-Gutwiller model. Section 3 considers periodic orbits with one single 2-antiparallel encounter. In the last section we consider periodic orbits with two 2-antiparallel encounters serial, with two 2-parallel encounter intertwined,
and with one 2-parallel encounter and one 2-antiparallel encounter intertwined. In each case, we prove the existence of partner orbits, estimate the action differences as well as provide symbolic dynamics for orbit pairs.

%

\medskip

\noindent 

\setcounter{equation}{0}
\section{The Hadamard-Gutzwiller model}
The Hadamard-Gutzwiller model is the geodesic flow on compact Riemann surfaces of constant negative curvature. 
It is well-known that any compact orientable surface 
with constant negative curvature is isometric to a factor $\Gamma\backslash \H^2$, 
where $\H^2=\{z=x+iy\in \C:\, y>0\}$ is the hyperbolic plane endowed  
with the hyperbolic metric $ds^2=\frac{dx^2+dy^2}{y^2}$ and
$\Gamma $ is a discrete subgroup of the projective Lie group $\PSL(2,\R)=\SL(2,\R)/\{\pm E_2\}$. The hyperbolic plane has constant Gaussian curvature $-1$. The group $\PSL(2,\R)$  acts transitively on $\H^2$ by 
M\"obius transformations
$z\mapsto \frac{az+b}{cz+d}$.
If the action has no fixed points, then the factor $\Gamma\backslash\H^2$  has a Riemann surface structure.
Such a surface is a closed Riemann surface of genus at least $2$ 
and has the hyperbolic plane $\H^2$ as the universal covering; so the natural projection $\pi_\Gamma: \H^2\rightarrow \Gamma\backslash\H ^2, \pi_\Gamma(z)=\Gamma z,\ z\in\H^2$ becomes a local isometry. This implies that $\Gamma\backslash\H^2$ also has constant curvature $-1$. 
The geodesic flow $(\varphi_t^\X)_{t\in \R}$ on the unit tangent bundle $\X=T^1(\Gamma\backslash\H^2)$
goes along the unit speed geodesics on $\Gamma\backslash\H^2$. 
On the other hand, the unit tangent bundle $T^1(\Gamma\backslash\H ^2)$
is isometric to the quotient space 
$\Gamma\backslash \PSL(2,\R)=\{\Gamma g,g\in\PSL(2,\R)\}$, 
which is the system of right co-sets of $\Gamma$ in $\PSL(2,\R)$, by an isometry
$\Xi$.
Then  the geodesic flow $(\varphi_t^\X)_{t\in\R}$ can be equivalently expressed as the natural 
``quotient flow'' $\varphi_t(\Gamma g)=\Gamma g a_t$ 
on $X=\Gamma\backslash\PSL(2,\R)$  associated to the flow $\varphi^\G_t(g)=g a_t$ on $\G:=\PSL(2,\R)$
by the conjugate relation 
\[\varphi_t^\X=\Xi^{-1}\circ\varphi_t\circ\Xi\quad \mbox{for all}\quad t\in\R.\]
Here $a_t\in\PSL(2,\R)$ denotes the equivalence class obtained from the matrix $A_t=\scriptsize\Big(\begin{array}{cc}
	e^{t/2} & 0\\ 0 & e^{-t/2}
\end{array}\Big)\in\SL(2,\R)$. 

\smallskip 
There are some more advantages to work on $X=\Gamma\backslash\PSL(2,\R)$
rather than on $\X=T^1(\Gamma\backslash\H^2)$. One can calculate explicitly the stable and unstable manifolds 
at a point $x=\Gamma g\in X$ to be
\[W^s_X(x)=\{\Gamma gb_s,s\in\R\}
\quad \mbox{and}\quad W^u_X(x)=\{\Gamma gc_u, u\in\R\},\]
where  $b_s=\{\pm B_s\},c_u=\{\pm C_u\}\in\PSL(2,\R)$ denote
the equivalence classes obtained from 
$B_s=\scriptsize\Big(\begin{array}{cc}1 &s\\ 0&1 \end{array} \Big), \ C_u=\scriptsize\Big(\begin{array}{cc}
	1&0\\ u&1
\end{array}\Big)\in\SL(2,\R)$. If the space $X$ is compact, then the flow $(\varphi_t)_{t\in\R}$
is a hyperbolic flow.

There is a natural Riemannian metric on $\G=\PSL(2,\R)$
such that the induced metric function $d_\G$ is left-invariant under $\G$. 
We define a metric function $d_{X}$ on $X=\Gamma\backslash\PSL(2,\R)$ by 
\[ d_{X}(x_1, x_2)
=\inf_{\gamma_1, \gamma_2\in\Gamma} d_{\G}(\gamma_1 g_1, \gamma_2 g_2)
=\inf_{\gamma\in\Gamma} d_{\G}(g_1, \gamma g_2), \]   
where $x_1=\Gamma g_1 $, $x_2=\Gamma g_2$.

General references for this section are \cite{bedkeanser,einsward}, 
and these works may be consulted for the proofs to all results which are stated above

\subsection{Poincar\'e sections}\label{Poincsec} 
It is well-known that the Riemann surface $\Gamma\backslash\H^2$ is compact if and only if 
the quotient space  $X=\Gamma\backslash\PSL(2,\R)$ is compact. 

First we recall the definitions of Poincar\'e sections in  \cite{HK,Huynh16}.
\begin{definition}\label{Poindn1}\rm 
	Let $x\in X$ and $\eps>0$. The {\em Poincar\'e sections} of radius $\eps$ at $x$ are defined by
	\begin{equation*}
		\P_\eps(x)=\{\Gamma g c_ub_s : |u|<\eps, |s| <\eps\},
	\end{equation*}
and \[\P'_\eps(x)=\{\Gamma gb_sc_u:|s|<\eps,\ |u|<\eps\},\]
	where $g\in \G$ is such that $x=\Gamma g$. 
	
\end{definition}
	If $z=\Gamma gc_ub_s\in \P_\eps(x)$ (resp. $z=\Gamma gb_sc_u\in \P'_\eps(x)$), we write
$z=(u,s)_x$ (resp. $z=(s,u)'_x$).
Note that the couple $(u,s)$ are not unique. As we will see below, if $X$ is compact and $\eps$ is sufficiently small, then the uniqueness of couple $(u,s)$ is obtained. 
\begin{lemma}\label{sigma_0}
	If the space $X=\Gamma\backslash\PSL(2,\R)$ is compact, then there exists $\sigma_0>0$ such that
	\begin{equation}\label{sigma0}d_\G(\gamma g, g)>\sigma_0\quad\mbox{for all}\quad \gamma\in\Gamma\setminus\{e\}. 
	\end{equation}
\end{lemma}
See \cite[Lemma 1, p. 237]{ratcliff} for a similar result on $\Gamma\backslash\H^2$.

	\begin{lemma}[\cite{Huynh16}] \label{sigma0u}If the space $X=\Gamma\backslash\PSL(2,\R)$ is compact
		and $\eps\in(0,\frac{\sigma_0}4)$, then
		for each $z\in \P_\eps(x)$ 
		there exist a unique couple $(u_z,s_z)\in (0,\eps)^2$ such that
		$z=\Gamma g c_{u_z} b_{s_z}$, where $g\in\PSL(2,\R)$ is such that $x=\Gamma g$, and we call $(u_z,s_z)$ the coordinates of $z$.  		
	\end{lemma}
%
%
%

\subsection{Conjugacy classes}
Let $\Gamma$ be a discrete subgroup of $\PSL(2,\R)$.

\begin{definition}\label{concla-def} \rm 
	(a) An element $\gamma\in\Gamma$ is called {\em primitive} if $\gamma=\zeta^m$ 
	for some $\zeta\in\Gamma$ implies that $m=1$ or $m=-1$. 
	
	(b) The {\em  conjugacy class} of $\gamma\in\Gamma$ is defined by
	\[ {\{\gamma\}}_\Gamma=\{\sigma\gamma\sigma^{-1}: \sigma\in\Gamma\}. \] 
	The collection of all conjugacy classes of primitive elements in $\Gamma\setminus\{e\}$  
	are denoted by ${\mathfrak C}_\Gamma$; here $e=[E_2]$ denotes the unity of $\PSL(2,\R)$.
\end{definition} 

	For $g=[G]\in\PSL(2,\R), G=\big({\scriptsize\begin{array}{cc}a&b\\c&d\end{array}} \big)\in\SL(2,\R)$,
the trace of $g$ is defined by $\tr(g)=|a+d|.$
If the action of $\Gamma$ on $\H^2$ is free and the factor $\Gamma\backslash\H^2$ is compact
then all elements $g\in\Gamma\setminus\{e\}$ are hyperbolic \cite[Theorem 6.6.6]{ratcliff}, i.e. $\tr(g)>2$.

Denote by ${\mathcal PO}_X$ the set of all periodic orbits of the flow $(\varphi_t)_{t\in\R}$. We define a mapping 
\begin{equation}\label{varsig} 
	\varsigma: {{\mathcal PO}}_X\to\mathfrak{C}_\Gamma
\end{equation} 
as follows. Take a periodic orbit $c$ of the flow, 
any point $x$ on $c$, and let $T>0$ be the prime period for $x$. 
Then $\varphi_T(x)=x$, and the definition of the flow 
implies that there are $g\in \PSL(2,\R)$ and $\gamma\in\Gamma$ such that $x=\Pi_\Gamma(g)$ 
and $\gamma=g a_T g^{-1}$, due to $\Gamma g a_T=xa_T=\varphi_T(x)=x=\Gamma g$; 
note that $\gamma\neq e$, since otherwise $a_T=e$ so that $T=0$. 
Then put $\varsigma(c)={\{\gamma\}}_\Gamma$. 
\begin{lemma}\label{cjo} Suppose that all elements in $\Gamma\backslash \{e\}$ are hyperbolic. 
	Then the mapping $\varsigma$ defined by \eqref{varsig} is a bijection between the periodic orbits ${\mathcal PO}_X$ of the flow $(\varphi_t)_{t\in\R}$ and the collection of all conjugacy classes of primitive elements ${\mathfrak{C}_\Gamma}$ in $\Gamma\setminus \{e\}$.
\end{lemma}

\subsection{Anosov closing lemma, connecting lemma}

The next two results are illustrated in Figure \ref{Anosovconnect}\,(a). For proofs, see \cite{HK,Huynh16}.
\begin{lemma}[{Anosov closing lemma I}]\label{anosov1}
	Suppose that $\eps\in (0, \frac{1}{4})$, 
	$x\in X$, $T\ge 1$, and $\varphi_T(x)\in {\mathcal P}_\eps(x)$. If $\varphi_T(x)={(u, s)}_x\in\P_\eps(x)$, 
	in the notation from Definition \ref{Poindn1}, then there are $x'=(\sigma,\eta)_x\in {\mathcal P}_{2\eps}(x)$ and $T'\in\R$ so that 
	\begin{equation*}\label{x'}\varphi_{T'}(x')=x'\quad \mbox{and}\quad d_X(\varphi_t(x),\varphi_t(x'))<2|u|+|\eta|< 4\eps\quad \mbox{for all}\quad t\in[0,T].
	\end{equation*}
	Furthermore,
	\begin{equation*}\label{T-T'1}
		\Big|\frac{T'-T}2-\ln(1+us)\Big| <5|us|e^{-T}
	\end{equation*}
	and
	\begin{eqnarray*}
		|\sigma|<2|u|e^{-T},\quad |\eta|<\frac{3|s|}2.
	\end{eqnarray*}
\end{lemma}

\begin{remark}\label{Alr}\rm 
	According to the proof of the Anosov closing lemma I in \cite[Theorem 2.3]{HK}, $x=\Gamma g, g\in\PSL(2,\R)$ and  $\zeta\in\Gamma$ is such that $ga_T=\zeta g c_ub_s$ then 
	$gc_{\sigma}b_\eta a_{T'}=\zeta g c_\sigma b_\eta$. This yields that the periodic orbit of $x'=\Gamma g c_\sigma b_\eta$ corresponds to the conjugacy class $\{\zeta\}_\Gamma$, provided that all elements in $\Gamma\backslash\{e\}$ are hyperbolic.
\end{remark}
Using the other version of Poincar\'e sections,  we have a respective  statement for the Anosov closing lemma which will be also useful afterwards.

\begin{lemma}[{Anosov closing lemma II}\empty]\label{anosov2}
	Suppose that $\eps\in\, (0, \frac{1}{4})$, 
	$x\in X$, $T\ge 1$, and $\varphi_T(x)\in \P'_\eps(x)$. If $\varphi_T(x)={(s, u)}'_x\in\P'_\eps(x)$, 
	in the notation from Definition \ref{Poindn1}, then there are $x'=(\eta,\sigma)'_x\in {\mathcal P}'_{2\eps}(x)$ and $T'\in\R$ so that 
	\begin{eqnarray*}\varphi_{T'}( x)= x\quad \mbox{and}\quad d_X(\varphi_t(x),\varphi_t(x'))\leq 2|u|+|\eta|<4\eps\quad \mbox{for all}\quad t\in[0,T].
	\end{eqnarray*}
	Furthermore,
	\begin{equation*}
		\Big|\frac{T'-T}2\Big| <4|us|e^{-T}
	\end{equation*}
	and
	\begin{equation*}
		|\sigma|<2|u| e^{-T},\quad |\eta|\leq \frac{3|s|}{2}.
	\end{equation*}
\end{lemma}

\begin{lemma}[{Connecting lemma}]\label{clm}
	Let $x_j\in X$ be $T_j$-periodic point of the flow $(\varphi_t)_{t\in\R}$ for $j=1,2$ and $T_1+T_2\geq 1$ and let $\eps\in(0,\frac14)$.	If $x_2=(u_1,s_1)_{x_1}\in \P_{\eps}(x_1)$, then there are
	$x\in X$ and $T>0$ such that $\varphi_T(x)=x$,
	\begin{eqnarray}\label{cl1} 
		d_X(\varphi_t(x),\varphi_t(x_1))& < &5\eps\quad \mbox{for all}\quad t\in[0,T_1],\\ \label{cl2} 
		d_X(\varphi_{t+T_1}(x),\varphi_t(x_2))&<&5\eps\quad \mbox{for all}\quad t\in[0,T_2],
	\end{eqnarray}
	and
	\begin{equation}
		\Big|\frac{T-(T_1+T_2)}{2}-\ln(1+us)\Big| <7|us|(e^{-T_1}+e^{-T_2}).
	\end{equation}
	Furthermore, if $x_i=\Gamma g_i$ for some $g_i\in\PSL(2,\R)$, then $x=\Gamma g_1 c_{u e^{-T_1}+\sigma} b_\eta$, where $\sigma,\eta\in\R$ satisfy
	\begin{equation}\label{cl3} 
		|\sigma|<2|u| e^{-T_1-T_2},\ \ |\eta|<\frac{3|s|}2
	\end{equation} 
and the orbit of $x$ corresponds to the conjugacy class $\{\gamma_1\gamma_2\}_\Gamma$, where $\gamma_1,\gamma_2\in\Gamma$ such that $g_1a_{T_1}=\gamma_1 g_1, g_2 a_{T_2}=\gamma_2g_2$,
provided that all elements in $\Gamma\backslash\{e\}$ are hyperbolic.
\end{lemma}
See Figure \ref{Anosovconnect}\,(b) for an illustration. 
\begin{proof} For the existence of a periodic point $x$ and \eqref{cl1}-\eqref{cl3}, see
	\cite[Theorem 2.6]{Huynh16}. For the last assertion, we can choose $g_1,g_2\in \PSL(2,\R)$ such that 
	$x=\Gamma g_1, y=\Gamma g_2$ and
	$g_2=g_1c_ub_s$. 
	According to the proof of the previous lemma  in \cite[Theorem 2.6]{Huynh16}, if $\gamma_1,\gamma_2\in\Gamma$ such that
	$g_1 a_{T_1}=\gamma_1 g_1, g_2 a_{T_2}=\gamma_2 g_2$ then 
	$g_1 c_{ue^{-T_1}+\sigma}b_\eta a_T=\gamma g_1 c_{ue^{-T_1}+\sigma}b_\eta$ with
	\begin{eqnarray*}
		\gamma&=& g_2 b_{-s} a_{-T_1}a_{T_1+T_2}b_{-s(1-e^{-T_2})}c_{-u(1-e^{-T_1})}a_{T_1}c_{-u}g_1^{-1}
		\\
		&=& g_2a_{T_2} b_{-se^{-T_2}}b_{-s(1-e^{-T_2})}c_{-u(1-e^{-T_1})}c_{-ue^{-T_1}}a_{T_1}g_1^{-1}\\
		&=& g_2 a_{T_2}b_{-s}c_{-u}g_1^{-1}\gamma_1\\
		&=&\gamma_2 g_2 b_{-s}c_{-u}g_1^{-1}\gamma_1
		\\
		&=&\gamma_2\gamma_1.
	\end{eqnarray*}
This implies that the orbit of $x$ corresponds to the conjugacy class 
\[\{\gamma_2\gamma_1\}_\Gamma=\{ \gamma_1\gamma_2\}_\Gamma.\]
\end{proof}

\begin{figure}[ht]
	\begin{center}
		\begin{minipage}{1\linewidth}
			\centering
			\includegraphics[angle=0,width=0.99\linewidth]{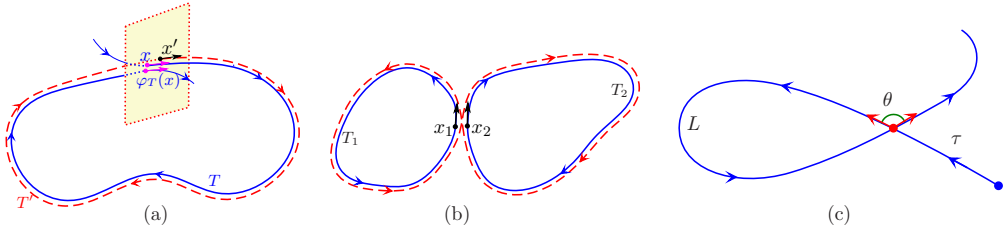}
		\end{minipage}
	\end{center}
	\caption{(a) Anosov closing lemma\ \ \ \ (b) Connecting lemma\ \ \ \ (c) Orbit with a self-crossing}\label{Anosovconnect}
\end{figure} 
\subsection{Self-crossings}
Recall that ${\mathcal X}=T^1(\Gamma\backslash\H^2)$ denotes the unit tangent bundle of the factor $\Gamma\backslash\H^2$; see Figure 
\ref{Anosovconnect}\,(c) for an illustration for the next result.
\begin{lemma}[Self-crossings,\cite{HK}]\label{selfcthm} 
	Suppose that all elements of $\Gamma\setminus\{e\}$ are hyperbolic 
	and let  $\tau\in\R,\, L>0,\,\theta\in\,(0,\pi)$, and ${\mathtt x}\in \X$ be given. 
	The orbit of ${\mathtt x}$ under the geodesic flow 
	${(\varphi^\X_t)}_{t\in\R}$ crosses itself in configuration space at the 
	time $\tau$, at the angle $\theta$, and creates a loop of length $L$ if and only if
	\begin{equation}\label{sc1} \mbox{either}\quad \Gamma ga_{\tau+ L}=\Gamma g a_\tau d_\theta \quad\mbox{or}\quad
		\Gamma g_{\tau+L}=\Gamma g a_{\tau} d_{-\theta} 
	\end{equation}
	holds for any $g\in\PSL(2,\R), \Gamma g=\Xi({\mathtt x})$.
	Furthermore, 
	\begin{equation}\label{sc2}
		e^{-L}<\cos^2\Big(\frac{\theta}{2}\Big). 
	\end{equation}
\end{lemma}

\subsection{Encounters and partner orbits}
\begin{definition}[Time reversal]\label{tr}\rm
	The {\em time reversal map} $\T:\X \rightarrow \X$ is defined by 
	\[\T(p,\xi)=(-p,\xi)\quad \mbox{for}\quad (p,\xi)\in \X.\]
	The respective time reversal map on $X=\Gamma\backslash\PSL(2,\R)$ is determined by 
	\begin{equation*}\label{revedn}
		\T(x)=\Gamma gd_\pi \quad \mbox{for}\quad  x=\Gamma g\in X,
	\end{equation*}
	where $d_\pi\in\PSL(2,\R)$ is the equivalence class of the matrix 
	$D_\pi=\scriptsize\Big(\begin{array}{cc} 0&1\\-1&0\end{array}\Big)\in{\rm SL(2,\R)}$.
\end{definition} 
Using Lemma \ref{adpi} below, we have
\begin{equation}\label{tr2}
	\varphi_t(\T(x))=\T(\varphi_{-t}(x))
	\quad \mbox{for}\quad x\in X\quad\mbox{and}\quad t\in\R.
\end{equation}

Next, we recall the notions of orbit pairs and partner orbits. Roughly speaking, two periodic orbits are called an orbit pair if they are close enough to each other in configuration space, not for the whole time, since otherwise they would be identical, but they decompose to the same number of parts and any part of one orbit is close to some part of the other. The following is a rigorous definition of  orbit pairs, which is recalled from \cite{Huynh16}. 
\begin{definition}[Orbit pair/Partner orbit]\label{partnerdf1}\rm 
	Let $\eps>0$ be given. 
	Two
	given $T$-periodic orbit $c$ and $T'$-periodic orbit $c'$ of the flow
	$(\varphi_t)_{t\in\R}$ are called an {\em $\eps$-orbit pair}
	if 
	there are $L\geq 2, L\in \Z$ and two decompositions  of $[0,T]$ and $[0,T']:$ $0=t_0<\cdots<t_L=T$
	and $0=t_0'<\cdots<t_L'=T'$, 
	and a permutation $\sigma:\{0,1,\dots,L-1\}\rightarrow 
	\{0,1,\dots, L-1\}$ such that 
	for each $j\in\{0,\dots,L-1\}$, 
	either
	\begin{equation*}
		d_X(\varphi_{t+t_j}(x), \varphi_{t+t_{\sigma(j)}'}(x'))<\eps
		\quad \mbox{for all}\quad t\in [0,t_{j+1}-t_j]
	\end{equation*}
	or
	\begin{equation*}
		d_X\Big(\varphi_{t+t_j}(x), \varphi_{t-t_{\sigma(j)+1}'}(\T(x'))\big)<\eps
		\quad \mbox{for all}\quad t\in [0,t_{j+1}-t_j]
	\end{equation*}
	holds for some $x\in c$ and $x'\in c'$.
	Then $c'$ is called an {\em $\eps$-partner orbit} of $c$ and vice versa. 
\end{definition}

\begin{definition}[Encounter]\rm Let $\eps>0$ and $L\in \Z, L\geq 2$ be given.
	We say that a periodic orbit $c$ of the flow $(\varphi_t)_{t\in\R}$ has an {\em$(L,\eps)$-encounter} if
	there are $x_1,\dots, x_L\in c$ such that for each $j\in\{2,\dots,L\}$,
	\[\mbox{either}\quad x_j\in \P_\eps(x_1)\quad \mbox{or}\quad \T(x_j)\in \P_\eps(x_1).\]
	The point $x_2,\dots,x_L$ are called {\em piercing points}. 
	If either $x_j\in\P_\eps(x_1)$ holds for all $2=1,\dots,L$ or $\T(x_j)\in\P_\eps(x)$ holds for all $j=2,\dots,L$  then
	the encounter is called {\em parallel encounter}; otherwise it is called {\em antiparallel encounter}.
\end{definition}

\subsection{Auxiliary results}

The next result is a decomposition of $\PSL(2,\R)$.
\begin{lemma}[\cite{HK}]\label{decompo}
	Let $g=[G]\in {\rm PSL}(2, \R)$ for $G=\scriptsize \Big(\begin{array}{cc} a & b \\ c & d\end{array}\Big)
	\in {\rm SL}(2, \R)$. 
	\begin{itemize}
		\item[(a)] If $a\neq 0$, then $g=c_u b_s a_t$ for 
		\begin{equation*}
			t=2\ln |a|,\quad s=ab,\quad u=\frac{c}{a}. 
		\end{equation*} 
		\item[(b)] If $d\neq 0$, then $g=b_s c_u a_t$ for 
		\[ t=-2\ln |d|,\quad s=\frac{b}{d},\quad u=cd. \] 
	\end{itemize} 
\end{lemma}  

\begin{lemma}\label{adpi}  The following relations hold for $t\in\R$: 
	\begin{equation}\label{dpi}a_t d_{\pi}=d_{\pi} a_{-t},\quad 
		b_t d_{\pi}=d_{\pi} c_{-t},\quad 
		c_t d_{\pi}=d_{\pi} b_{-t}.
	\end{equation} 
\end{lemma}
\begin{proof} In ${\rm SL}(2, \R)$ we calculate 
	\begin{eqnarray*} 
		A_t D_{\pi} & = & \Bigg(\begin{array}{cc} e^{t/2} & 0 \\
			0 & e^{-t/2}\end{array}\Bigg)\Bigg(\begin{array}{cc} 0 & 1 \\
			-1 & 0\end{array}\Bigg)=\Bigg(\begin{array}{cc} 0 & e^{t/2} \\
			-e^{-t/2} & 0\end{array}\Bigg)\\
		&=&\Bigg(\begin{array}{cc} 0 & 1 \\
			-1 & 0\end{array}\Bigg)\Bigg(\begin{array}{cc} e^{-t/2} & 0 \\
			0 & e^{t/2}\end{array}\Bigg)
		=  D_{\pi} A_{-t} 
	\end{eqnarray*} 
	which upon projection yields the first one. The argument is analogous for the others.
	
\end{proof}

Owing to the hyperbolicity, the flow $(\varphi_t)_{t\in\R}$ is expansive, i.e.,  two orbits cannot stay too close together without being identical; see \cite[Lemma 1.5]{bowen}. 
 For periodic orbits, we have the following property; see \cite[Theorem 3.14]{HK} for a proof.

\begin{lemma}\label{per-coinc} Let $X=\Gamma\backslash {\rm PSL}(2, \R)$ be compact. 
	Then there is $\eps_\ast>0$ with the following property. 
	If $\eps\in (0, \eps_\ast)$ and if $x_1, x_2\in X$ 
	are periodic points of ${(\varphi_t^X)}_{t\in\R}$ having the periods $T_1, T_2>0$ 
	such that $|T_1-T_2|\le\sqrt 2\eps$ and 
	\[ d_X(\varphi_t^X(x_1), \varphi_t^X(x_2))<\eps
	\quad\mbox{for all}\quad t\in [0, \min\{T_1, T_2\}], \]  
	then $T_1=T_2$ and the orbits of $x_1$ and $x_2$ under ${(\varphi_t^X)}_{t\in\R}$ are identical. 
\end{lemma}

\section{Periodic orbits with one single 2-antiparallel encounter}
Let us first recall from \cite{Huynh17} periodic orbits with one single 2-antiparallel encounter. It was shown that a given periodic orbit including  one single 2-antiparallel encounter has a partner orbit. The action difference between the orbit pair is estimated with an exponentially small error bound. Periodic orbits having small-angle self-crossing are special cases of this phenomenon.
The results in this section will be applied for the main results in Subsection \ref{4.1sub}.

\begin{theorem}\label{2anti} Suppose that $X=\Gamma\backslash\PSL(2,\R)$ is compact and let $\eps\in(0,\frac{\sigma_0}{8})$. 
	If a periodic orbit $c$ of the flow $(\varphi_t)_{t\in\R}$ on $X$
	with period $T>1$ has a $(2,\eps)$-antiparallel encounter, then it has a partner. Furthermore,
	let $x,y\in c$, $\T(y)=(u,s)_x\in\P_\eps(x)$ and $\varphi_{T_1}(x)=y, 0< T_1<T$. Then the partner is $\eps'$-partner with $\eps'=\eps+2(|u-se^{-T_1}|+|s-ue^{T_1-T}|)<8\eps$ and the action
	difference between the orbit pair satisfies
	
	\begin{equation}\label{TT'}
		\Big|\frac{T'-T}{2}-\ln(1+us)|<12\eps^2 (e^{-T_1}+ e^{T_1-T}),
	\end{equation}
	where $T'$ is the period of the partner. If $\eps \in (0,\frac{\eps_*}{18})$, then the partner orbit is unique.
\end{theorem}

\begin{proof} For the existence of a $T'$-periodic point $v$, whose orbit under the flow $(\varphi_t)_{t\in\R}$ is a $\eps'$-partner orbit, see \cite[Theorem 9]{Huynh17}. 
	It was shown that the action difference satisfies
	\begin{equation*}
		\Big|\frac{T'-T}{2}-\ln\big(1+(u-se^{-T_1})(s-ue^{T_1-T})\big)\Big|\leq |(u-se^{-T_1})(s-ue^{T_1-T})|e^{-T},
	\end{equation*}
which implies \eqref{TT'}. 
For the last assertion, it follows from \eqref{TT'} that $|T-T'|<30\eps^2$. 
Suppose that there is another partner orbit which has the same property, 
i.e., it is also $8\eps$-close to the original one and
its period called $T''$ satisfies $|T''-T|<30\eps^2$. 
Then these two partner orbits are $\eps_*$-close to each other for the whole time
and their periods satisfy
\[|T''-T'|\leq |T''-T|+|T'-T|\leq 60\eps^2 <\sqrt 2 \eps_*. \]
By Lemma \ref{per-coinc}, the partner orbits must coincide. 
\end{proof}

\begin{figure}[ht]
	\begin{center}
		\begin{minipage}{1\linewidth}
			\centering
			\includegraphics[angle=0,width=0.7\linewidth]{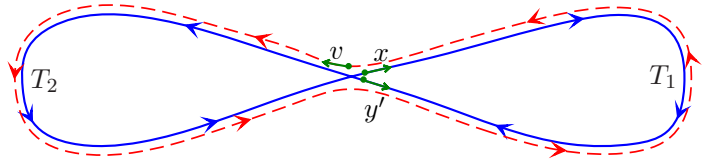}
		\end{minipage}
	\end{center}	\caption{Periodic orbit with a single 2-antiparallel encounter has a partner orbit.}\label{2ape}
\end{figure} 
\begin{remark}\label{2antirm}\rm According to the proofs of \cite[Theorem 9]{Huynh17} and the Anosov closing lemma I, we have
 $v=\Gamma ga_{-T_2}c_{ue^{-T_2}+\sigma}b_\eta$, where $T_2=T_1-T,\sigma,\eta\in\R$ satisfy
	\begin{equation}
		|\eta-\tilde s|<2\tilde s^2|\tilde u|+2|\tilde s|e^{-T}
\quad		\mbox{and}\quad |\sigma|<2|\tilde u| e^{-T}
	\end{equation}
with $\tilde s=u-se^{-T_1}$ and $\tilde u=s-ue^{-T_2}$.
{\hfill$\diamond$}\end{remark}
The next result is a new view of symbolic dynamics.
\begin{theorem}[Symbolic dynamics]\label{ccrm}
	In the setting of Theorem \ref{2anti}, let 	$g,h\in\PSL(2,\R)$ be such that $x=\Gamma g, y=\Gamma h$ and $hd_\pi=h'=gc_ub_s$.
	  If $\gamma_1,\gamma_2\in\Gamma$ such that
	$\gamma_1=ga_{T_1}h^{-1}$ and $\gamma_2=ha_{T_2}g^{-1}$,
	where $T_2=T-T_1$,
	then the original orbit corresponds to the conjugacy class $\{\gamma_1\gamma_2 \}_\Gamma$ and
	the partner orbit corresponds to the conjugacy class $\{\gamma_1^{-1}\gamma_2 \}_\Gamma$.
\end{theorem}
\begin{proof}Note that due to $\varphi_{T_1}(x)=y$ and $\varphi_{T_2}(y)=x$, there are $\gamma_1,\gamma_2\in\Gamma$ such that
	$\gamma_1=ga_{T_1}h^{-1}$ and $\gamma_2=ha_{T_2}g^{-1}$ as assumption. 
	Since $x=\Gamma g$ is a $T$-periodic orbit, we have 
	\[ \gamma =ga_Tg^{-1} \]
	for some $\gamma\in\Gamma$,
	and by Lemma \ref{cjo}, the orbit through $x$ (called  $c$) 
	corresponds to the conjugacy class ${\{\gamma\}}_\Gamma$:
	\[ \gamma_1\gamma_2=(ga_{T_1}h^{-1})(ha_{T_2}g^{-1})=ga_{T}g^{-1}=\gamma. \]
	This means that the orbit $c$ corresponds to the conjugacy class ${\{\gamma\}}_\Gamma
	={\{\gamma_1\gamma_2\}}_\Gamma$. Next, 	let  $\zeta\in\Gamma$ be such that 
	\[\zeta=\bar ga_T b_{-s'}c_{-u'} (\bar g)^{-1} \]
	for $\bar g=gc_ua_{-T_2}=h'b_{-s}a_{T_2}$. According to the proof of Theorem \ref{2anti} and Remark \ref{Alr},
	the partner orbit corresponds to the conjugacy class $\{\zeta\}_\Gamma$.  Now, 	
	\begin{eqnarray*}
		\zeta & = &
		h'b_{-s}a_{-T_2}a_T b_{-(u-s^{-T_1})}c_{-(s-ue^{-T_2})}a_{T_2}c_{-u}g^{-1}
		\\
		&=& h' a_{T_1} b_{-s e^{-T_1}}b_{-(u-s^{-T_1})}c_{-(s-ue^{-T_2})}
		c_{-ue^{-T_2}} a_{T_2}g^{-1}
		\\
		&=& h'a_{T_1}b_{-u}c_{-s}a_{T_2}g^{-1}
		\\
		&=&\gamma_1^{-1}\gamma_2,
	\end{eqnarray*}
	noting
	$h'a_{T_1}=\gamma_1^{-1}g'$ owing to $ga_{T_1}=\gamma_1 h$. Therefore,
	the partner orbit $c'$ corresponds to the conjugacy class $\{\zeta\}_\Gamma=\{\gamma_1^{-1}\gamma_2\}_\Gamma$, completing the proof.
\end{proof}

\section{Periodic orbits including 2 encounters responsible for the third order term}

Let us first review orbit pairs responsible for the cubic contribution to $K(\tau)$.
Note that a sufficiently long periodic orbit has a huge number of self-encounters which may involve arbitrarily many orbit stretches. Heusler, M\"uller et al. \cite{HMBH,mueller2005} show that only orbit pairs differing in two 2-encounters or in one single 3-encounter are responsible for the third order.  There are two ways of connections of orbit stretches forming two  2-encounters, namely serial and intertwined, whereas the two stretches of each encounter may be either close in phase space
(depicted by nearly parallel arrows $\psr$ or $\doublearrow[3]$), which is called {\em parallel encounter}, or
almost mutually time-reversed (like in $\asr$ or $\doublearrow[2]$), which is called {\em antiparallel encounter}.
In each way, therefore, there are three possibilities: both 2-encounters  are parallel-encounters
,
one 2-parallel encounter and one 2-antiparallel encounter, and
two 2-antiparallel encounters, i.e. there are totally six cases. 
Only three of them lead to (genuine) periodic orbits: two 2-antiparallel encounters serial, one parallel-encounter and one antiparallel encounters intertwined and two anti-parallel encounters intertwined,
which are responsible in the cubic order to the form factor.  The others form pseudo-periodic orbits and do not contribute to the spectral form factor.  

In this section we only consider periodic orbits including two 2-antiparallel encounters contributing to the third term of the spectral form factor. The case of one 3-parallel encounter is rigorously done in \cite[Section 3.3]{Huynh16}. Orbits with a single 3-antiparallel encounter can be done analogously. 

Throughout this section, we assume that the space $\Gamma\backslash\PSL(2,\R)$ is compact. 

\subsection{Periodic orbits including two 2-antiparallel encounters serial} \label{4.1sub}

In this subsection we consider periodic orbits having two 2-antiparallel encounters serial, which is so called {\em antiparallel-antiparallel serial} ({\em aas} for short) in \cite{mueller2005}. 
Periodic orbits with either two small-angle self-crossings in configuration space ($\doublearrow[2]$ and $\doublearrow[2]$) or
one small-angle self-crossing and one anti-parallel avoided self-crossing ($\doublearrow[2]$ and $\asr$), or
two anti-parallel avoided self-crossings ($\asr$ and  $\asr$) in configuration space are special cases of aas.

\begin{theorem}\label{aasthm} Let $\eps\in(0,\frac{\eps_0}{24})$ and let $c$ be  a $T$-periodic orbit of the flow $(\varphi_t)_{t\in\R}$ on $X$
	 having two $(2,\eps)$-antiparallel encounters serial. More precisely, 	let $x,y,z,w\in c$ and $T_1,T_2,T_3,T_4>0$ be such that $T_1+T_2+T_3+T_4=T$,  
	$\varphi_{T_1}(x)=y, \varphi_{T_2}(y)=z,
	\varphi_{T_3}(z)=w$ and $\varphi_{T_4}(w)=x$,
	and ${\mathcal T}(z)=(u_1,s_1)_y\in \P_\eps(y ), {\mathcal T}(w)=(u_2,s_2)_x\in \P_\eps(x)$ satisfying 
	\begin{eqnarray}\label{cdt0}
		|u_1|>6\eps e^{-T_2},\quad 
		|s_1|> 30\eps^3+13\eps e^{-T_1}+5\eps e^{-T_2}+3\eps e^{-T_3}.
	\end{eqnarray} 
	Then $c$ has a $20\eps$-partner orbit which differs in both encounters of the original orbit and  the action difference satisfies
\begin{equation}\label{T'-Tassthm}
	\Big| \frac{T'-T}{2}-\ln(1+u_1s_1)(1+u_2s_2)\Big|< \eps^2( 21e^{-T_1}+30 e^{-T_2}+12 e^{-T_3}+19 e^{-T_4}).
\end{equation} 
In addition, if $\eps\in (0,\frac{\eps_*}{40})$, then the partner is unique.
\end{theorem}

\begin{proof} 
	The sketch of the proof is as follows; see Figure \ref{aasf} for an illustration. First, we apply Theorem \ref{2anti} for the left encounter to have one partner orbit, which is depicted by the dashed line in Figure \ref{aasf}\,(b). Next, we show that the new orbit admits one 2-antiparallel encounter; see Figure \ref{aasf}\,(c). Finally, we apply Theorem  \ref{2anti} again to get a partner orbit depicted by the dashed line in Figure \ref{aasf}\,(d).	
	
	\begin{figure}[ht]
		\begin{center}
			\begin{minipage}{1\linewidth}
				\centering
				\includegraphics[angle=0,width=1\linewidth]{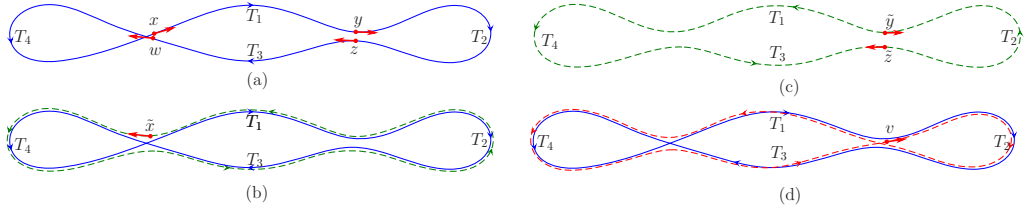}
			\end{minipage}
		\end{center}	\caption{Construction of partner orbit for a given periodic orbit with two 2-antiparallel encounters serial.}\label{aasf}
	\end{figure} 
	
Let $x=\Gamma g, y=\Gamma h, z=\Gamma k, w=\Gamma l$ for some
	$g,h,k,l\in\PSL(2,\R)$ and set $g'=gd_\pi, h'=hd_\pi, k'=kd_\pi, l'=ld_\pi$. 
	By hypothesis, $\varphi_{T_{123}}(x)=w$ and ${\mathcal T}(w)=(s_2,u_2)_{x}\in\P_\eps(x)$, where $T_{123}=T_1+T_2+T_3$. 
	Due to Theorem \ref{2anti}, 
	there are $\tilde x=\Gamma l' b_{-s_2}a_{-T_4}c_{\sigma_2}b_{\eta_2}=\Gamma g c_{u_2}a_{-T_4} c_{\sigma_2}b_{\eta_2}\in X$ and $\widetilde T>0$ such that
	\begin{equation}\label{TT'1}
		\Big|\frac{\widetilde  T-T}{2}-\ln(1+u_2s_2)\Big|\leq 12\eps^2 (e^{-T_4}+e^{-T_{123}}).
	\end{equation}
	Furthermore, 
	\begin{eqnarray}\label{9eps}
		d_X(\varphi_t(\tilde x), \varphi_t(w))< 7\eps\quad \mbox{for}\quad t\in [0, T_{4}]
	\end{eqnarray}
	and
	\begin{equation}\label{9eps1}
		d_X(\varphi_t(\tilde x), \varphi_t(x'))< 7\eps\quad \mbox{for}\quad t\in [T_{4},T]. 
	\end{equation}
	
	Next we show that the orbit of $\tilde x$ possesses one 2-antiparallel encounter.  
	We write
	\begin{eqnarray*}
		\T(\tilde z):=\varphi_{T_4+T_3}(\tilde x)&=&\Gamma l' b_{-s_2}a_{-T_4}c_{\sigma_2}b_{\eta_2} a_{T_4+T_3}\\
			&=&\Gamma l'b_{-s_2}a_{T_3} c_{\sigma_2 e^{T_3+T_4}} b_{\eta_2 e^{-T_3-T_4}}\\
		&=&\Gamma l'a_{T_3} b_{-s_2e^{-T_3}}c_{\sigma_2 e^{T_3+T_4}} b_{\eta_2 e^{-T_3-T_4}}\\
		&=&\Gamma k'b_{-s_2e^{-T_3}}c_{\sigma_2 e^{T_3+T_4}} b_{\eta_2 e^{-T_3-T_4}},
	\end{eqnarray*}
	using $\varphi_{T_3}(\T(w))=\T(z)$ and $b_{s}a_t=a_tb_{se^{-t}}$, $c_u a_t=a_t c_{ue^t}$ for all $u,s,t\in\R$. 
	Then 
	\begin{equation}\label{tildez}\tilde z=\Gamma k c_{s_2e^{-T_3}}b_{-\sigma_2 e^{T_3+T_4}} c_{-\eta_2 e^{-T_3-T_4}}; 
	\end{equation} recall \eqref{tr2}. 
	Now,
	\begin{eqnarray*}
		\T(\hat y)&:=&\varphi_{T_4+T_3+T_2}(\tilde x)
		= \Gamma h'b_{-s_2e^{-T_2-T_3}}c_{\sigma_2 e^{T_2+T_3+T_4}} b_{\eta_2 e^{-T_2-T_3-T_4}}\\
		&=& (\Gamma k c_{s_1}b_{u_1})b_{-s_2e^{-T_2-T_3}}c_{\sigma_2 e^{T_2+T_3+T_4}} b_{\eta_2 e^{-T_2-T_3-T_4}}\\
		&=& (\Gamma k c_{s_2e^{-T_3}}b_{-\sigma_2 e^{T_3+T_4}} c_{-\eta_2 e^{-T_3-T_4}})\\
		&&
		c_{\eta_2 e^{-T_3-T_4}}
		b_{\sigma_2 e^{T_3+T_4}}c_{-s_2e^{-T_3}+s_1}b_{u_1-s_2e^{-T_2-T_3}}c_{\sigma_2 e^{T_2+T_3+T_4}} b_{\eta_2 e^{-T_2-T_3-T_4}}\\
		&=&	(\Gamma k c_{s_2e^{-T_3}}b_{-\sigma_2 e^{T_3+T_4}} c_{-\eta_2 e^{-T_3-T_4}})
		c_{\tilde u_1}b_{\tilde s_1}a_{\tilde \tau_1},
	\end{eqnarray*}
	where	
	\begin{eqnarray}\notag 
		\tilde u_1 &=& s_1+\eta_2 e^{-T_3-T_4}  -s_2e^{-T_3}+\sigma_2 e^{T_2+T_3+T_4}+\frac{1}{1+\tilde\rho_1}\times  
		\\ \notag 
		&&\quad\  \times\,\big((s_1-s_2e^{-T_3})\sigma_2 e^{T_2+T_3+T_4}(u_1-s_2e^{-T_2-T_3})\\
		&&\quad \ -(s_1-s_2e^{-T_3}+\sigma_2 e^{T_2+T_3+T_4})\tilde\rho_1\big), \label{tildeu1}
		\\ \notag 
		\tilde	s_1 &=& u_1+\sigma_2 e^{T_3+T_4}-s_2e^{-T_2-T_3}+\eta_2 e^{-T_2-T_3-T_4}
		\\ \notag 
		&&\quad\,+\,\tilde\rho_1 \big((2+\tilde\rho_1)\eta_2 e^{-T_2-T_3-T_4}+\sigma_2 e^{T_3+T_4}+u_1-s_2e^{-T_2-T_3}\big)\\ \label{tildes1}
		&&\quad\,+\,
		(s_1-s_2e^{-T_3})\sigma_2 e^{T_3+T_4}(u_1-s_2e^{-T_2-T_3})(1+\tilde\rho_1),
		\\
		\tilde	\tau_1 &=& 2\ln(1+\tilde\rho_1),\label{tildetau}
	\end{eqnarray}
	with  \[\tilde\rho_1=\sigma_2 e^{T_2+T_3+T_4}(\sigma_2 e^{T_3+T_4}+u_1-s_2e^{-T_2-T_3})
	+(s_1-s_2e^{-T_3})\sigma_2 e^{T_3+T_4}\big(1+\sigma_2 e^{T_2+T_3+T_4}(u_1-s_2e^{-T_2-T_3})\big),\]
	owing to Lemma \ref{decompo}\,(a).
	A short calculation shows that  $|\tilde u_1|<2\eps$ and
	$|\tilde s_1|<2\eps$. 
	This means that $$\T(\tilde y)=\varphi_{-\tilde\tau_1}(\T(\hat y))=(\tilde u_1,\tilde s_1)_{\tilde z}\in \P_{2\eps}(\tilde z),$$ where $\tilde y:=\varphi_{\tilde \tau_1} (\hat y)$. 
	Apply Theorem \ref{2anti} to obtain $v=\tilde z a_{-\widetilde T_2}c_{\tilde u_1e^{-\widetilde T_2}+\sigma_1} b_{\eta_1}\in X$ and $T'\in \R$ such that 
	\begin{equation}\label{TT'2}
		\Big|\frac{T'-\widetilde T}{2}-\ln(1+ \tilde u_1\tilde s_1)\Big|\leq 28\eps^2(e^{-T_1-T_3-T_4+\tilde\tau_1}+e^{T_1+T_3+T_4-\tilde\tau_1-\widetilde T})
		<7\eps^2e^{-T_4}+30\eps^2 e^{-T_2}
	\end{equation}
	and \begin{eqnarray}\label{eta1hats1}
		|\eta_1-\hat s_1|&<&2\hat s_1^2|\hat u_1|+2|\hat s_1| e^{-\widetilde T}<30\eps^3+10\eps e^{-\widetilde T},\\
		|\sigma_1|&<& 6\eps e^{-\widetilde T}, 
	\end{eqnarray}
where 	$\hat u_1= \tilde s_1-\tilde u_1 e^{\widetilde T_2-\widetilde T}, \hat s_1=\tilde u_1-\tilde s_1 e^{-\widetilde T_2}$.
Furthermore, 
	\begin{eqnarray*}
		d_X(\varphi_t(v), \varphi_t(\tilde y))< 13\eps\quad \mbox{for}\quad t\in [0, T_{2}]
	\end{eqnarray*}
	and 
	\begin{equation*}
		d_X(\varphi_t(v), \varphi_t(\T(\tilde z)))< 13\eps\quad \mbox{for}\quad  t\in [T_{2},\widetilde T].
	\end{equation*}
	This means that the orbit of $v$ is $13\eps$-close to the orbit of $\tilde x$. Recalling from \eqref{9eps} and \eqref{9eps1} that the orbit of $\tilde x$ is $7\eps$-close to the orbit of $x$, we deduce that the orbit of $v$ is $20\eps$-close to the original one.

	Next, in order to establish an estimate for the action difference, observe that by \eqref{tildeu1}, \eqref{tildes1}, 	
	\begin{equation}\label{tildeus-}
			|\tilde u_1-s_1|<7\eps e^{-T_1}+2 \eps e^{-T_3},\quad
			|\tilde s_1-u_1|< \eps e^{-T_3}
	\end{equation}
	after a short calculation. This yields
	\[|\ln(1+\tilde u_1\tilde s_1)-\ln(1+u_1s_1)|\leq 21\eps^2e^{-T_1}+11\eps^2 e^{-T_3}   \] and hence 
	\begin{eqnarray}\label{T'tildeT2} 
		\Big| \frac{T'-\widetilde T}{2}-\ln(1+u_1s_1)\Big| \leq 21\eps^2 e^{-T_1}+30\eps^2 e^{-T_2}+11\eps^2 e^{-T_3}+7\eps^2 e^{-T_4},
	\end{eqnarray} 
	using \eqref{TT'2}. The estimate \eqref{T'-Tassthm} follows from \eqref{TT'1} and \eqref{T'tildeT2}.
	
	Next we are going to show that the partner orbit is different from the original one. For, we find a point in the partner orbit which lies in the Poincar\'e section of $y$ and is different from $y$ and ${\cal T}(z)$. Letting $\widetilde T_2=T_2-\tilde \tau_1$, we have 
	$\varphi_{\widetilde T-\widetilde T_2}(\tilde z)=\tilde y$ and $\varphi_{\widetilde T_2}(\tilde y)=\tilde z$. 
Recall $v=\tilde z a_{-\widetilde T_2}c_{\tilde u_1e^{-\widetilde T_2}+\sigma_1} b_{\eta_1}$. Using
	\eqref{tildez} and Lemma \ref{decompo}\,(a), we write
	\begin{eqnarray*}
		v&=&\tilde z a_{-\widetilde T_2}c_{\tilde u_1e^{-\widetilde T_2}+\sigma_1} b_{\eta_1}
		\\
		&=&\Gamma k c_{s_2e^{-T_3}}b_{-\sigma_2 e^{T_3+T_4}} c_{-\eta_2 e^{-T_3-T_4}}a_{-T_2+\tilde\tau_1}c_{\tilde u_1e^{-\widetilde T_2}+\sigma_1} b_{\eta_1}\\
		&=&(\Gamma ka_{-T_2})c_{s_2e^{-T_2-T_3}}
		b_{-\sigma_2 e^{T_2+T_3+T_4}} c_{-\eta_2 e^{-T_2-T_3-T_4}+\tilde u_1e^{-T_2}+\sigma_1 e^{-\tilde\tau_1}} b_{\eta_1 e^{\tilde\tau_1}} a_{\tilde\tau_1}
		\\
		&=&\Gamma h c_{u_v}b_{s_v}a_{\tau_v},
	\end{eqnarray*} 
	where
	\begin{eqnarray*}
		u_v&=&s_2e^{-T_2-T_3}+\frac{-\eta_2 e^{-T_2-T_3-T_4}+\tilde u_1e^{-T_2}+\sigma_1 e^{-\tilde\tau_1}}{1+\rho_v},\\
		s_v&=&(-\sigma_2 e^{T_2+T_3+T_4}+\eta_1 e^{\tilde\tau_1}+\eta_1 e^{\tilde\tau_1}\rho_v)(1+\rho_v),\\
		\tau_v&=&2\ln(1+\rho_v),
	\end{eqnarray*}
here
	\[\rho_v=(-\sigma_2 e^{T_2+T_3+T_4})(-\eta_2 e^{-T_2-T_3-T_4}+\tilde u_2e^{-T_2}+\sigma_1 e^{-\tilde\tau_1});\]
	recall $\tilde\tau_1$ from \eqref{tildetau}.
	A short calculation shows that 
	\begin{equation}\label{uvsv}
		|u_v|< 6\eps e^{-T_2},
		|s_v-\eta_1|<6\eps e^{-T_1},
	\end{equation}
which imply that $\tilde v=\varphi_{\tau_v}(z)=(u_v,s_v)_y\in \P_{6\eps}(y)$. By assumption, $\eps<\frac{\sigma_0}{24}$ and $\T(z)=(u_1,s_1)_y$
with $$|u_1|\geq 6\eps e^{-T_2}>|u_v|.$$
As a consequence of Lemma \ref{sigma0u}, we get $\tilde v\ne \T(z)$. It remains to check that $\tilde v\ne y$. 
	Recalling that $\hat s_1=\tilde u_1-\tilde s_1 e^{-\widetilde T_2}$, 
		it follows from \eqref{eta1hats1} and \eqref{tildeus-}  that
	\begin{eqnarray}\notag 
		|\eta_1-s_1|
		&\leq& |\eta_1-\hat s_1|+|\hat s_1-\tilde u_1|+|\tilde u_1-s_1| \\ \label{eu1}
		&<& 30\eps^3+7\eps e^{-T_1}+5\eps e^{-T_2}+3\eps e^{-T_3}.
	\end{eqnarray}	
Using  \eqref{uvsv}, \eqref{eu1}, together with the assumption \eqref{cdt0} we deduce
	\[|s_v|\geq |s_1|- |s_v-\eta_1|-|\eta_1-s_1|>0,\]
	which shows $\tilde v\ne y$. Consequently, the orbit of $v$ is different from the obit of $x$ and its time reverse. 
	
For the last assertion,  suppose that there is another $20\eps$-partner orbit which differs in both encounters of the original orbit. Then the two partner orbits are $40\eps$-close to each other for the whole time and the period difference $|T''-T'|<50\eps^2$, where $T''$ is the period of the new partner. Due to $\eps<\frac{\eps_*}{40}$, the two partners are 
$\eps_*$-close to each other for the whole time and 
$|T''-T'|<\eps_*$; so  they must be identical by Lemma \ref{per-coinc}. The proof is complete.
\end{proof}

\begin{remark} \rm (a)
	The orbit of $\tilde x$ is also a partner orbit of the original one, which is depicted by the dotted  line in Figure \ref{aasf}\,(b). However, this partner orbit differs only in one encounter and this orbit pair only contribute to the second order term of the spectral form factor as a Sieber-Richter pair.
	
	(b) If we first apply Theorem \ref{2anti} for the other encounter, then we  have another partner orbit for the original one. This orbit pair also contribute to the second order term of
	$K(\tau)$. Then, using the same argument of the previous proof, we get the same partner orbit.
	 
\end{remark}
The next result provides a new view of symbolic dynamics of orbit pair in the preceding theorem. 
\begin{theorem}\label{sblthm}In the setting of Theorem \ref{aasthm}, let $\gamma_1,\gamma_2,\gamma_3,\gamma_4\in\Gamma$ be such that $ga_{T_1}=\gamma_1 h, ha_{T_2}=\gamma_2 k, ka_{T_3}=\gamma_3 l, la_{T_4}=\gamma_4g$. 
	Then the original orbit and the partner orbit correspond to the conjugacy classes $\{\gamma_1\gamma_2\gamma_3\gamma_4\}_\Gamma$ and	$\{\gamma_1^{-1}\gamma_4\gamma_3^{-1}\gamma_2 \}_\Gamma$, respectively.
\end{theorem}

\begin{proof} First observe that $a_{T_1}=g^{-1}\gamma_1 h,
	a_{T_2}=h^{-1}\gamma_2 k, a_{T_3}=k^{-1}\gamma_3l,a_{T_4}=l^{-1}\gamma_4 g$ leads to
	\[ga_T=g a_{T_1}a_{T_2} a_{T_3}a_{T_4}
	=g(g^{-1}\gamma_1 h)(h^{-1}\gamma_2 k)(k^{-1}\gamma_3l)(l^{-1}\gamma_4 g)
	=\gamma_1\gamma_2\gamma_3\gamma_4 g.\]
	This yields the orbit of $x=\Gamma g$ under the flow $(\varphi_t)_{t\in\R}$, which is the original orbit, corresponds to the conjugacy class $\{\gamma_1\gamma_2\gamma_3\gamma_4\}_\Gamma$ by the definition of the mapping $\varsigma$ in \eqref{varsig}. 
	
	Next,  due to $ga_{T_1+T_2+T_3}=\gamma_1\gamma_2\gamma_3l$
	and $la_{T_4} =\gamma_4g$,  the orbit of $\hat x$ corresponds to the 
	conjugacy class \begin{eqnarray*} 
		\{\zeta\}_\Gamma=\{(\gamma_1\gamma_2\gamma_3)^{-1}\gamma_4 \}_\Gamma=\{\gamma_3^{-1}\gamma_2^{-1}\gamma_1^{-1}\gamma_4\}_\Gamma
		=\{\gamma_2^{-1}\gamma_1^{-1}\gamma_4\gamma_3^{-1}\}_\Gamma.
	\end{eqnarray*} 
by Theorem \ref{ccrm}. Similarly, the orbit of $v$ corresponds to the conjugacy class
	\begin{eqnarray*}
		\{\xi\}_\Gamma =\{({\gamma_2^{-1}})^{-1} \gamma_1^{-1}\gamma_4\gamma_3^{-1}\}_\Gamma
		=\{\gamma_2\gamma_1^{-1}\gamma_4\gamma_3^{-1} \}_\Gamma=\{\gamma_1^{-1}\gamma_4\gamma_3^{-1}\gamma_2\}_\Gamma,
	\end{eqnarray*} 
as was to be shown.
\end{proof} 
The following corollary considers periodic orbits of the geodesic flow having two small-angle self-crossings. 
\begin{corollary}
Suppose that all elements in $\Gamma\setminus\{e\}$ are hyperbolic. Suppose that $T$-periodic orbit of the geodesic flow ${(\varphi^\X_t)}_{t\in\R}$ on $\X=T^1(\Gamma\backslash\H^2)$
crosses itself in configuration space at a time $T_1$, at an  angle $\theta_1$, creates a loop of length $T_2$ and then crosses itself again after a time $T_3$ at an angle $\theta_2$,
and creates another loop of length $T_4$ with $T_1+T_2+T_3+T_4=T$.
If  $0<\phi<\frac13$
for $\phi=\max\{\pi-\theta_1,\pi-\theta_2\}$, 
then it has a $36|\sin(\phi/2)|$-partner orbit; the period of the partner orbit denoted by $T'$ satisfies 
\begin{align}\notag
	\Big|\frac{T'-T}{2}-&\ln\left(1+\sin^2(\phi_1/2)\right)\left(1+\sin^2(\phi_2/2)\right)\Big|\\
	&<\sin^2(\phi/2)(21 e^{-T_1}+31 e^{-T_2}+13 e^{-T_3}+19 e^{-T_4}).\label{T'-Tcr}
\end{align}
Furthermore, if $\Gamma\backslash\H^2$ is compact and $\phi<\frac{\eps_*}{20}$ then the partner is unique. 
\end{corollary}
\begin{proof}
	
	Let  ${\mathtt x}\in {\mathcal X}=T^1(\Gamma\backslash\H^2)$ 
	be a $T$-periodic point of the flow $(\varphi_t^\X)_{t\in\R}$ and let
	\begin{equation}\label{seya}
		\varphi_{T_1}^{{\mathcal X}}({\mathtt x})={\mathtt y},
		\quad\varphi_{T_2}^{{\mathcal X}}({\mathtt y})={\mathtt z},\quad \varphi_{T_3}^{{\mathcal X}}({\mathtt z})={\mathtt w},
		\quad\varphi_{T_4}^{{\mathcal X}}({\mathtt w})={\mathtt x}.
	\end{equation} 	
Let $\phi_1=\pi-\theta_1, \phi_2=\pi-\theta_2$, $\phi=\max\{\phi_1,\phi_2\}$ and  assume that $|\phi|<\frac13$. Then in particular 
	\begin{equation}\label{sinphi/2} 
		\Big|\sin\Big(\frac{\phi_i}{2}\Big)\Big|\le\frac{|\phi_i|}{2}
		<\frac16
	\end{equation} 
	holds for $i=1,2$. 
	Set $x=\Xi({\mathtt x}), y=\Xi ({\mathtt y}), z=\Xi({\mathtt z})$, and $w=\Xi({\mathtt w})$; recall the isometry $\Xi: T^1(\Gamma\backslash\H^2)\to \Gamma\backslash\PSL(2,\R)$. It follows from Theorem  \ref{selfcthm} that
	\begin{equation*}
		\mbox{either}\quad \Gamma h=\Gamma k d_{\theta_1} \quad\mbox{or}\quad \Gamma h=\Gamma kd_{-\theta_1}
	\end{equation*}
and 
\begin{equation*}
	\mbox{either}\quad \Gamma l=\Gamma g d_{\theta_2}\quad \mbox{or}\quad\Gamma l=\Gamma g d_{-\theta_2}
\end{equation*}
	with some $g,h,k,l\in\PSL(2,\R)$ such that $\Gamma g=x,\Gamma h=y, \Gamma k=z$ and $\Gamma l=w$.  We only consider 
	$\Gamma h=\Gamma k d_{\theta_1}$ and $\Gamma l=\Gamma g d_{\theta_2}$, the other cases are similar. Define 
	$z'={\mathcal T}(z)$ and $w'={\mathcal T} (w)$; recall the notation ${\mathcal T}$ from Definition \ref{tr}. Then
	\[z'=\Gamma h d_{\pi-\theta_1}=\Gamma hd_{\phi_1} \quad \mbox{and}\quad w'=\Gamma g d_{\pi-\theta_2}=\Gamma g d_{\phi_2}.\]
	For $i=1,2$, apply Lemma \ref{decompo}\,(a) to write 
	\begin{equation}\label{dexpress} 
		d_{\phi_i}=c_{u_i}b_{s_i} a_{\tau_i},
	\end{equation} 
	where 
	\begin{equation}\label{uisi} \tau_i=2\ln(\cos(\phi_i/2)),\quad u_i=\tan(\phi/2),\quad s_i=-\sin(\phi_i/2)\cos(\phi_i/2).
	\end{equation} 
	Owing to (\ref{sinphi/2}), we have 
	\begin{equation}\label{cosphi/2}
		\cos\Big(\frac{\phi_i}{2}\Big)>\frac{5}{6}.
	\end{equation} 
Define $\eps=\frac{6}{5}|\sin(\phi/2)$. Observe that 
	\begin{equation}\label{tausu-bd}  
		|u_i|=|\tan(\phi_i/2)|\le \frac65|\sin(\phi_i/2)|\leq \eps,
		\ |s_i|=|\sin(\phi_i/2)\cos(\phi_i/2)|\le |\sin(\phi_i/2)|<\eps,
	\end{equation} 
	and 
	\[ |\tau_i|=|\ln(1-\sin^2(\phi_i/2))|\le 2\sin^2(\phi_i/2)
	\le \frac12\, \eps^2, \]
	due to $|\ln(1+z)|\le 2|z|$ for $|z|\le 1/2$. Denote 
	$\tilde{z}=\varphi_{\tau_1}(z)$ and 
	$\tilde{w}=\varphi_{\tau_2}(w)$. This leads to
	\begin{align*}
		{\mathcal T}(\tilde z)&=\varphi_{-\tau_1}(z')=\Gamma k' a_{-\tau_1}=\Gamma h c_{u_1}b_{s_1}=(u_1,s_1)_y\in \P_\eps(y),
		\\
		{\mathcal T}(\tilde w)&=\varphi_{-\tau_2}(w')=\Gamma l' a_{-\tau_2}=\Gamma g  c_{u_2}b_{s_2}=(u_2,s_2)_x\in \P_\eps(x),
\end{align*}
	using \eqref{tr2}. Define $\widetilde T_2=T_2+\tau_1, \tilde T_3=T_3-\tau_1+\tau_2,
	\widetilde T_4=T_4-\tau_2,\widetilde T_1=T_1$. Then
	$T_1+\widetilde T_2+\widetilde T_3+\widetilde T_4=T$, 
	$\varphi_{T_1}(x)=y,
	\varphi_{\widetilde T_2}(y)=\tilde z,
	\varphi_{\widetilde T_3}(z)=\tilde w$, and 
	$\varphi_{\widetilde T_4}(\tilde w)=x$.  
	  We now apply Theorem \ref{2anti} with $\eps=\frac{6}{5}|\sin(\phi/2)|$  to have a partner orbit, which is  $33|\sin(\phi/2)|$-close to the original one and has period $T'$ satisfying 	  
	  \begin{align*}
	  	\Big| \frac{T'-T}{2}-\ln(1+\sin^2(\phi_1/2))(1+\sin^2(\phi_2/2))\Big|&< \eps^2(21 e^{-T_1}+30 e^{-\widetilde T_2}+12 e^{-\widetilde T_3}+19 e^{-\widetilde T_4})\\
	  	&<\eps^2(21 e^{-T_1}+31 e^{-T_2}+13e^{-T_3}+19 e^{-T_4}),
	  \end{align*} 
  which is
	  \eqref{T'-Tcr}.
	
	For the last assertion, if $\phi<\frac{\eps_*}{12}$,
	 then $\eps=\frac{6}{5}|\sin(\phi/2)|<\frac{\eps_*}{20}$ and whence the partner orbit is unique according to Theorem \ref{aasthm}; see Figure \ref{2sacfg} for an illustration.
\end{proof}

	\begin{figure}[ht]
	\begin{center}
		\begin{minipage}{1\linewidth}
			\centering
			\includegraphics[angle=0,width=0.8\linewidth]{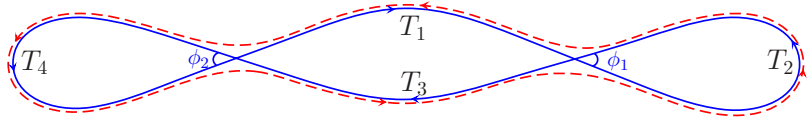}
		\end{minipage}
	\end{center}	\caption{A periodic orbit with 2 small-angle-self-crossings has a partner orbit which has 2 avoided crossings. }\label{2sacfg}
\end{figure}

\begin{remark}\rm
	The existence of the partner orbit in the preceding theorem  does not need the condition that the space is compact.  In fact, according \cite[Theorem 3.11]{HK} and the proof of Theorem \ref{aasthm}, the partner orbit has a smaller period $T'<T$.
\end{remark}

\subsection{Periodic orbits including two 2-parallel encounters intertwined}
In this subsection we consider periodic orbits with two 2-parallel encounters intertwined, which is so called {\em parallel-parallel intertwined} ({\em ppi} for short) in \cite{mueller2005}.

\begin{theorem}\label{ppithm}Let $\eps\in(0,\frac{\sigma_0}{20})$. 
		Suppose that a $T$-periodic orbit $c$ of the flow $(\varphi_t)_{t\in\R}$ on $X$
		with period $T$ has two $(2,\eps)$-parallel encounters intertwined. For instance, suppose that there are  $x,y,z,w\in c$ and $T_1,T_2,T_3,T_4>0$  such that $T_1+T_2+T_3+T_4=T$,
			$\varphi_{T_1}(x)=y, \varphi_{T_2}(y)=z,
		\varphi_{T_3}(z)=w$ and $\varphi_{T_4}(w)=x$
		with $z=(u_1,s_1)_x\in \P_\eps(x ),  w=(u_2,s_2)_y\in \P_\eps(y)$ satisfying
		\begin{eqnarray}\label{cdt}
			|u_2|>9\eps e^{-T_4},\
		|s_2|>72\eps^3+5\eps e^{-T_1}+2\eps e^{-T_3}.
		\end{eqnarray}
		Then $c$ has a $19\eps$-partner orbit of period $T'$ which differs in both encounters 	and the action difference satisfies
		\begin{equation}\label{T'-Tppithm}
			\Big|\frac{T'-T}{2}-\ln(1+s_1u_1)(1+s_2u_2)|\leq  54\eps^4+25\eps^2(e^{-T_1}+e^{-T_2}+ e^{-T_3}+e^{-T_4}) . 
		\end{equation} 
		If  $\eps\in (0,\frac{\eps_*}{38}) $, then the partner orbit is unique. 
	\end{theorem}

\begin{proof}
	The construction of a partner orbit is summarized as follows. We first apply the Anosov closing lemmas for the first encounter to obtain two shorter periodic orbits, which are expressed by dashed line and dotted line in Figure \ref{ppif}\,(b).
	We next show that the obtained orbits create a pseudo-orbit (see Figure \ref{ppif}\,(c)). Finally we use the connecting lemma to get a new periodic orbit, which is illustrated by the dashed line in Figure \ref{ppif}\,(d).
	
	\begin{figure}[ht]
		\begin{center}
			\begin{minipage}{1\linewidth}
				\centering
				\includegraphics[angle=0,width=0.95\linewidth]{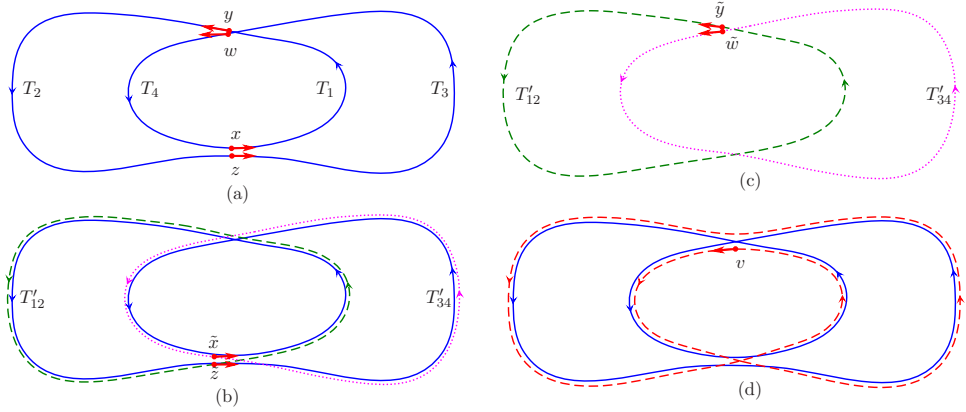}
			\end{minipage}
		\end{center}
		\caption{Reconnection to give the partner orbit for orbit with two parallel encounters intertwined.}\label{ppif}
	\end{figure} 

Denote $T_{12}=T_1+T_2$ and $T_{34}=T_3+T_4$.	By assumption, $\varphi_{T_{12}}(x)=z=(u_1,s_1)_x\in \P_\eps(x)$. According to the Anosov closing lemma I, there exist $\tilde x=\Gamma hb_{\eta_x}c_{\sigma_x}\in X$ and $T_{12}'\in \R$ 
	such that $\varphi_{T'_{12}}(\tilde x)=\tilde x$,
	\begin{equation}\label{T12pc} 
		\Big|\frac{T_{12}'-T_{12}}{2}-\ln(1+u_1s_1) \Big| \leq 5\eps^2 e^{-T_{1}-T_2} 
	\end{equation}
	and 
	\begin{equation}\label{tildexx}
		d_X(\varphi_t(\tilde x),\varphi_t(x))<4\eps,\quad\mbox{for all}\quad t\in [0,T_1+T_2].
	\end{equation}
The orbit of $\tilde x$ is depicted by the dashed line in Figure \ref{ppif}\,(b). Observe that
$\varphi_{T_{34}}(z)=x=(-s_1,-u_1)_{z}'\in \P'_\eps(z)$.
Apply the Anosov closing lemma II, there exist $\tilde z=\Gamma l b_{\eta_z}c_{\sigma_z}\in X$ and
$T_{34}'\in\R$ such that $\varphi_{T_{34}'}(\tilde z)=\tilde z$,
\begin{equation}\label{T34pc}
	\Big|\frac{T_{34}'-T_{34}}{2}\Big|\leq 4\eps^2 e^{-T_{3}-T_4}
\end{equation}
and
	\begin{equation}\label{tildezz}
	d_X(\varphi_t(\tilde z),\varphi_t(z))<4\eps,\quad\mbox{for all}\quad t\in [0,T_3+T_4].
\end{equation}
The orbit of $\tilde z$ is depicted by the dotted line in Figure \ref{ppif}\,(b). 

Next we show that two orbits of $\tilde x$ and $\tilde z$ form a pseudo-orbit, and hence  it is possible to connect them to obtain a longer periodic orbit. Using Lemma \ref{decompo}\,(b), we write
\begin{eqnarray*}
	\varphi_{T_3}(\tilde z)=\Gamma lb_{\eta_z}c_{\sigma_z}a_{T_3}
	&=&\Gamma la_{T_3}b_{\eta_z e^{-T_3}}c_{\sigma_ze^{T_3}}  
	=\Gamma k b_{\eta_z e^{-T_3}}c_{\sigma_ze^{T_3}}\\
	&=&\Gamma h c_{u_2}b_{s_2}b_{\eta_z e^{-T_3}}c_{\sigma_ze^{T_3}}\\
	&=&(\Gamma hc_{\sigma_x e^{T_1}}b_{\eta_xe^{-T_1}})b_{-\eta_xe^{-T_1}}
	c_{u_2-\sigma_x e^{T_1}}b_{s_2+\eta_z e^{-T_3}}c_{\sigma_ze^{T_3}}\\
	&=&(\Gamma hc_{\sigma_x e^{T_1}}b_{\eta_xe^{-T_1}})b_{s}c_{u}a_{\tau},
\end{eqnarray*}
where
\begin{eqnarray}\notag 
	s&=&\eta_ze^{-T_1}+\frac{s_2+\eta_ze^{-T_3}}{ 1+ (u_2-\sigma_x e^{T_1})(s_2+\eta_ze^{-T_3})},\\ \notag 
	u&=& \left(u_2-\sigma_x e^{T_1}+\sigma_ze^{T_3}+(u_2-\sigma_x e^{T_1})(s_2+\eta_ze^{-T_3})\sigma_ze^{T_3}\right)\left(1+(u_2-\sigma_xe^{T_1})(s_2+\eta_ze^{-T_3})\right),\\ \label{tau}
 \tau&=&-2\ln\left(1+(u_2-\sigma_xe^{T_1})(s_2+\eta_ze^{-T_3})\right).
\end{eqnarray}
A short calculation shows that $|u|<3\eps, |s|<3\eps$. If we set  
\begin{eqnarray*}
	\tilde y:= \varphi_{T_1}(\tilde x)=\Gamma g c_{\sigma_x}b_{\eta_x}a_{T_1}
	=\Gamma g a_{T_1} c_{\sigma_x e^{T_1}}b_{\eta_xe^{-T_1}}
	=\Gamma hc_{\sigma_x e^{T_1}}b_{\eta_xe^{-T_1}}
\end{eqnarray*} 
and $\tilde w:=\varphi_{T_3-\tau}(\tilde z)$, then $\tilde y=(-u,-s)_{\tilde w}\in \P_{3\eps}(\tilde w)$.
Apply the connecting lemma I to obtain a $T'$-periodic point $v=\tilde w c_{-u e^{-T_{12}'}+\sigma}b_\eta$ satisfying
\begin{eqnarray}\label{T'1234}
	\Big| \frac{T'-(T'_{12}+T'_{34})}{2}-\ln(1+us)\Big|< 63\eps^2(e^{-T'_{12}}+e^{-T'_{34}})
\end{eqnarray}
and 
\begin{equation}\label{etas}
	|\eta+s|<2s^2|u|+2|s|e^{-T_1-T_2}<54\eps^3+6\eps e^{-T_1-T_2}.
\end{equation}
Furthermore,
\begin{equation*}
	d_X(\varphi_t(v), \varphi_t(\tilde z))<15\eps\quad \mbox{for}\quad t\in [0,T'_{3,4} ]
\end{equation*}
and 
\begin{equation*}
	d_X(\varphi_{t+T'_{34}}(v),\varphi_t(\tilde w))<15\eps
	\quad \mbox{for}\quad t\in [ 0,T'_{1,2}].
\end{equation*}
This yields that the orbit of $v$ is $15\eps$-close to the orbits of $\tilde x$ and $\tilde z$. 
Together with \eqref{tildexx} and \eqref{tildezz}, this implies that the orbit of $v$ is $19\eps$-close to the orbit of $x$. 

Next, we estimate the action difference. 
According to the Anosov closing lemmas, $$|\eta_x|<\frac{3\eps}2, |\eta_z|<\frac{3\eps}2, |\sigma_x|<2\eps e^{-T_1-T_2}, |\sigma_z|< 2\eps e^{-T_{3}-T_4}.$$ Observe that
\begin{equation}\label{s-s2}|u-u_2|< 9\eps^3+2\eps e^{-T_2}+2\eps e^{-T_4}\quad  \mbox{and}\quad |s-s_2|< 3\eps^3+2\eps e^{-T_1}+2\eps e^{-T_3}
\end{equation} 
imply
\[|\ln(1+us)-\ln(1+u_2s_2)|<54\eps^4+18\eps^2 e^{-T_1}+6\eps^2e^{-T_2}+18\eps^2 e^{-T_3}+6\eps^2 e^{-T_4}, \]
and hence
\begin{align}\notag 
	\Big|\frac{T'-T'_{12}-T'_{34}}{2}&-\ln(1+u_2s_2)\Big|\\
	&\leq 54\eps^4+24\eps^2 e^{-T_1}+24\eps^2e^{-T_2}+24\eps^2 e^{-T_3}+24\eps^2 e^{-T_4}, \label{T'pc}
\end{align}
owing to \eqref{T'1234}. The estimate \eqref{T'-Tppithm} follows from \eqref{T12pc}, \eqref{T34pc} and \eqref{T'pc}.

Next, we are going to show that the orbit of $v$ is different from the orbit of $x$. Analogously to the proof of Theorem \ref{aasthm}, we write
\begin{eqnarray*}
	v&=& \tilde w c_{ue^{-T_{34}'}+\sigma}b_\eta\\
	 &=& (\Gamma hc_{\sigma_x e^{T_1}}b_{\eta_xe^{-T_1}})b_{s}c_{u}c_{ue^{-T_{34}'}+\sigma}b_\eta\\
	 &=&\Gamma hc_{\sigma_x e^{T_1}}b_{\eta_xe^{-T_1}+s}c_{u+ue^{-T_{34}'}+\sigma}b_\eta\\
	 &=&\Gamma h c_{u_v}b_{s_v}a_{\tau_v}\\
	 &=&\Gamma h c_{u_2}b_{s_2}b_{\eta_z e^{-T_3}}c_{\sigma_ze^{T_3}}a_{-\tau} c_{ue^{-T_{34}'}+\sigma}b_\eta\\
	 &=&\Gamma h c_{u_2}b_{s_2+\eta_z e^{-T_3}}c_{\sigma_ze^{T_3}+ue^{-T_{34}'+\tau}+\sigma}b_{\eta e^{-\tau}}a_{-\tau}
	 \\
	 &=&\Gamma h c_{u_v}b_{s_v}a_{\tau_v}  ,
\end{eqnarray*}
where \begin{eqnarray*}
	u_v&=&u_2+\frac{\sigma_ze^{T_3}+ue^{-T_{34}'+\tau}+\sigma}{1+(s_2+\eta_z e^{-T_3})(\sigma_ze^{T_3}+ue^{-T_{34}'+\tau}+\sigma)},\\
	s_v&=&(s_2+\eta_z e^{-T_3}+\eta e^{-\tau}+\eta e^{-\tau}\rho_v)(1+\rho_v),\\
	\tau_v&=&2\ln\big(1+\rho_v\big)-\tau,
\end{eqnarray*}
recalling $\tau$ from \eqref{tau}; here
\[\rho_v=(s_2+\eta_z e^{-T_3})(\sigma_ze^{T_3}+ue^{-T_{34}'+\tau}+\sigma). \] A short calculation shows that $|s_v|<5\eps$ and $|u_v|<5\eps$.
Consequently, $\tilde v:=\varphi_{-\tau_v}(v)=(u_v,s_v)_y\in \P_{5\eps}(y)$. We need to check that $\tilde v\ne y$ and $\tilde v\ne w$. For, observe that
\begin{eqnarray*}
	|u_v-u_2|< 9\eps e^{-T_4},\quad 
|s_v-s_2-\eta|<	15\eps^3.
\end{eqnarray*}
Furthermore, it follows from \eqref{etas} and \eqref{s-s2} that 
\[|\eta+s_2|\leq |\eta+s|+|s-s_2|< 57\eps^3+5\eps e^{-T_1}+2\eps e^{-T_3}.\]
Using assumption \eqref{cdt}, we derive
\begin{eqnarray*}
	|u_v|&>&|u_2|-|u_v-u_2|>|u_2|-9\eps e^{-T_4}>0\\
	|s_v-s_2|&>& |s_2|-|s_v-s_2-\eta|-|\eta+s_2|>0.
\end{eqnarray*}
This means $u_v\ne 0$ and $s_v\ne s_2$. Owing to $\eps<\frac{\sigma_0}{20}$, we get $\tilde v \ne y, \tilde v\ne w$ and hence the orbit of $v$ is different from the orbit of $x$.

For the last assertion, recall that the orbit of $v$ is 19$\eps$-close to the original orbit. The uniqueness of partner orbit can be done analogously to the proof of Theorem \ref{aasthm}. 
\end{proof}

\begin{remark}\rm 
(a) 	The period of the partner orbit can be explicitly computed according to the proof of Anosov closing lemmas and the connecting lemma. The term $33\eps^4$ appears when we change the coordinates $(u,s)$ to $(u_2,s_2)$ (see \eqref{s-s2}) and it cannot avoid. 

(b) Condition \eqref{cdt} can be replaced by \begin{eqnarray}
	|u_1|>9\eps e^{-T_3},\ 
	|s_1|>72\eps^3+5\eps e^{-T_4}+2\eps e^{-T_2}. 
\end{eqnarray}
	
\end{remark}
\begin{theorem}\label{2picj}
	In the setting of Theorem \ref{ppithm}, suppose that $x=\Gamma g, y=\Gamma h, z=\Gamma k, w=\Gamma l$ for some $g,h,k,l\in\PSL(2,\R)$. If $\gamma_1,\gamma_2,\gamma_3,\gamma_4\in\Gamma$ satisfy 
	$ga_{T_1}=\gamma_1 h, ha_{T_2}=\gamma_2 k, ka_{T_3}=\gamma_3 l, la_{T_4}=\gamma_4g$, 
	then the original orbit corresponds to the conjugacy class $\{\gamma_1\gamma_2\gamma_3\gamma_4\}$ and
	the partner orbit corresponds to the conjugacy class
$\{\gamma_2\gamma_1\gamma_4\gamma_3 \}$.
\end{theorem}
\begin{proof}  That the original orbit corresponds to the conjugacy class $\{\gamma_1\gamma_2\gamma_3\gamma_4 \}$ can be done analogously to Theorem \ref{sblthm}. For the last assertion, we can choose $g,h,k,l\in \PSL(2,\R)$ such that 
	$k=gc_{u_1}b_{s_1}$ and $l=hc_{u_2}b_{s_2}$.
	Then $\varphi_{T_1+T_2}(\Gamma g)=\Gamma g c_{u_1}b_{s_1}$ and 
	$ga_{T_1+T_2}=\gamma_1 h a_{T_2}=\gamma_1\gamma _2 k=\gamma_1\gamma_2 gc_{u_1}b_{s_1}$. Similarly, 
	$\varphi_{T_3+T_4}(\Gamma k)=\Gamma k b_{-s_1}c_{-u_1}$ and 
	$k a_{T_3+T_4}=\gamma_3 la_{T_4}=\gamma_3\gamma_4 g=\gamma_3\gamma_4 kb_{-s_1}c_{-u_1}$. 
	By the Remark \ref{Alr}, the orbit of $\tilde x$ corresponds to the conjugacy class $\{\gamma_1\gamma_2\}_\Gamma=\{\gamma_2\gamma_1 \}_\Gamma$ and the orbit of $\tilde z$ corresponds to the class $\{\gamma_3\gamma_4\}_\Gamma=\{\gamma_4\gamma_3\}_\Gamma$. 
	The  partner orbit corresponds to the conjugacy class 
	\[\{\gamma_2\gamma_1\gamma_4\gamma_3\}_\Gamma; \]
	recall the last assertion of Lemma \ref{clm}.
\end{proof}
\subsection{Periodic orbits including one 2-antiparallel encounter and one 2-parallel encounter intertwined}

This subsection deals with periodic orbits with one 2-antiparallel encounter and one 2-parallel encounter intertwined, which is so called {\em antiparallel-parallel intertwined} ({\em api} for short) in \cite{mueller2005}.
\begin{figure}[ht]
	\begin{center}
		\begin{minipage}{1\linewidth}
			\centering
			\includegraphics[angle=0,width=0.95\linewidth]{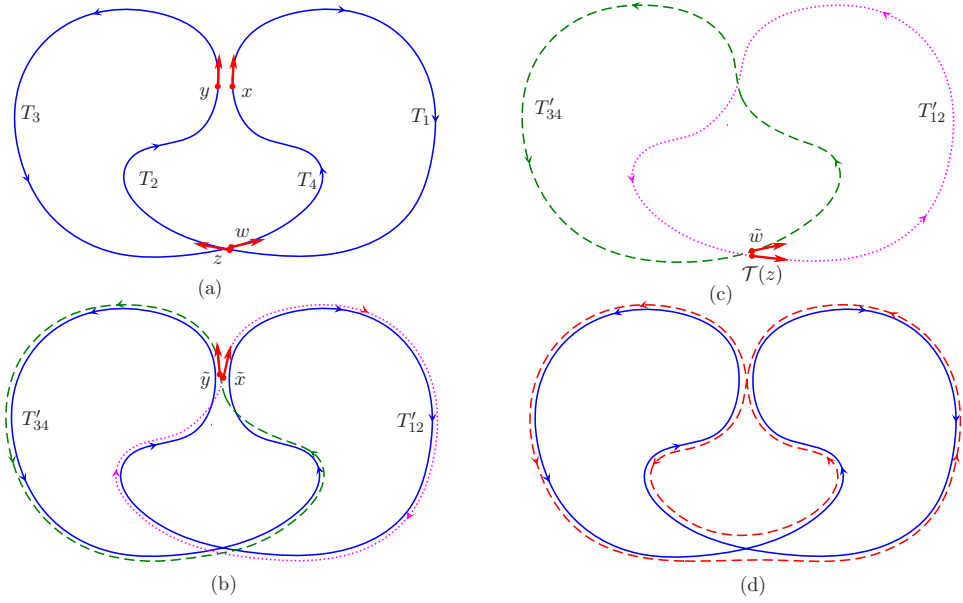}
		\end{minipage}
	\end{center}
	\caption{Reconnecting encounter stretches to form a partner orbit in api case. }\label{apif}
\end{figure} 

\begin{theorem} Let $\eps\in(0,\frac{\sigma_0}{20})$. 
	Suppose that a periodic orbit $c$ of the flow $(\varphi_t)_{t\in\R}$ on $X$
	with period $T>1$ has one $(2,\eps)$-parallel encounter and
	one $(2,\eps)$-antiparallel encounter intertwined. More precisely, suppose that there are  $x,y,z,w\in c$ and $T_1,T_2,T_3,T_4>0$ be such that $T_1+T_2+T_3+T_4=T$,
	$\varphi_{T_1}(x)=y, \varphi_{T_2}(y)=z,
	\varphi_{T_3}(z)=w$ and $\varphi_{T_4}(w)=x$
	and $z=(u_1,s_1)_x\in \P_\eps(x ), \T(w)=(u_2,s_2)_y\in \P_\eps(y)$ and
	\begin{eqnarray}\label{cdt2} 
		|u_1|>9\eps e^{-T_3},\quad 
		|s_1|>72\eps^3+5\eps e^{-T_4}+2\eps e^{-T_2}.
	\end{eqnarray}
	Then $c$ has a $19\eps$-partner orbit of period $T'$ which differs in both encounters and  the action difference satisfies
\[\Big|\frac{T'-T}{2}-\ln(1+s_1u_1)(1+s_2u_2)|\leq  54\eps^4+25\eps^2(e^{-T_1}+ e^{-T_3}+ e^{-T_4}) . \]
If  $\eps\in (0,\frac{\eps_*}{38})$, then the partner orbit is unique. 
\end{theorem}

\begin{proof}
The proof of this theorem is similar to that of Theorem \ref{ppithm}.	 Note that after obtaining two shorter periodic orbits, the dashed  and dotted  lines, we consider the time reversal of the former one (Figure \ref{apif}\,(c)), and apply the connecting lemma to have a partner depicted by dashed line in Figure \ref{apif}\,(d).
\end{proof}

\begin{remark}\rm Analogously to Theorem \ref{2picj}, in the setting of the previous theorem, if the original orbit corresponds to the conjugacy class $\{\gamma_1\gamma_2\gamma_3\gamma_4 \}_\Gamma$, then 
	the partner orbit corresponds to the conjugacy class
	\[\{(\gamma_2\gamma_1)^{-1}\gamma_4\gamma_3\}_\Gamma
	=\{\gamma_1^{-1}\gamma_2^{-1}\gamma_4\gamma_3  \}_\Gamma.\]
\end{remark}

\begin{remark}\rm 
	(a) The assumption of compactness is unnecessary for the existence of partner orbits in all cases above. However, we need it for the uniqueness. 
	
	(b) The conditions expressed by the coordinates of the piercing points \eqref{cdt0}, \eqref{cdt}, \eqref{cdt2} guarantee that the encounter stretches are separated by non-vanishing loops and whence the partner orbit and the original orbit do not coincide. A similar condition is needed in physics literature; see \cite{muellerthesis,mueller2005}. 
	
	(c) The approach in the present paper can be applied 
	to consider periodic orbits responsible for all order in $\tau$ to the spectral form factor $K(\tau)$.
\end{remark}


\begin{thebibliography}{99}
	\bibitem{bedkeanser}
	{T. Bedford, M. Keane, and C. Series (Eds.),}
	{\em\,\,Ergodic Theory, Symbolic Dynamics and Hyperbolic Spaces},
	Oxford University Press, Oxford 1991.
	\bibitem{berry}
	{ M.V. Berry,} ``Semiclassical theory of spectral rigidity'', 
	{\em Proc.~Roy.~Soc.~London Ser.~A}{\bf\,\,400}, 229-251 (1985). 
	
	
	\bibitem{bowen}{R. Bowen,} ``Symbolic dynamics for hyperbolic flows'',
	{\em Amer. J. Math.} {\bf 95(2)} (1973),  429-460.
	
	
	\bibitem{braun2002}
	{P. Braun, S. Heusler, S. M\"uller, and F. Haake,}
	``Statistics of self-crossings and avoided crossings of periodic orbits in the Hadamard-Gutzwiller model'',
	{\em Eur.~Phys.~J.~B}{\bf\,\,30}, 189-206 (2002).
 \bibitem{efetov}
 {K. Efetov,}{\em\,\,Supersymmetry in Disorder and Chaos},
 Cambridge University Press, Cambridge 1997. 
	\bibitem{einsward}
	{M. Einsiedler and T. Ward,}{\em\,\,Ergodic Theory With a View
		Towards Number Theory}, Springer, Berlin-New York 2011
	\bibitem{GO} {B. Gutkin and V.A. Osipov,} {``Clustering of periodic orbits in chaotic systems''}, {\em Nonlinearity} {\bf 2}, 177-199 (2012).
	\bibitem{haake}
	{F. Haake,}{\em\,\, Quantum Signatures of Chaos}, 3rd edition, Springer, Berlin-New York 2010.
	\bibitem{hoda} 
	{J.H. Hannay and A.M. Ozorio de Almeida,} ``Periodic orbits and a correlation function 
	for the semiclassical density of state'', 
	{\em J.~Phys.~A: Math.~Gen.}{\bf\,\,17}, 3429-3440 (1984).
	
	\bibitem{HMABH} {H. Heusler, S. M\"uller, A.  Altland, P.  Braun, and F. Haake,} 
	{``Periodic
		orbit theory of level correlations''}, {\em Phys. Rev. Lett.}{\bf\,\,98}, 044103 (2007).
	\bibitem{HMBH} {H. Heusler, S. M\"uller, P. Braun, and F. Haake,} {``Universal spectral form
		factor for chaotic dynamics''}, {\em J. Phys. A}{\bf\,\, 37}, L31-L37 (2004).
	

	\bibitem{HK}
	{H.M. Huynh and M. Kunze,} ``Partner orbits and action differences on compact factors of the hyperbolic plane.  I: Sieber-Richter pairs'', {\em Nonlinearity}{\bf\,\, 28}, 593-623 (2015).
	
	\bibitem{Huynh16} 
	{H.M. Huynh,} ``Partner orbits and action differences on compact factors of the hyperbolic plane.
	II: Higher-order encounters'', {\em Physica D.} {\bf\,\, 314}, 35-53 (2016).
	\bibitem{Huynh17}{H.M. Huynh,} ``On 2-antiparallel encounters on factors of the hyperbolic plane'', {\em Quy Nhon J. Sci.} {\bf 11(1)}, 5-15 (2017).
	\bibitem{KeatRob} 
	{J.P. Keating and J.M. Robbins,} ``Discrete symmetries and spectral statistics'', 
	{\em J.~Phys.~A: Math.~Gen.}{\bf\,\,30}, L177-L181 (1997) 
	\bibitem{muellerthesis}
	{S. M\"uller,}{\em\,\,Periodic-Orbit Approach to Universality in Quantum Chaos}, PhD thesis, 
	Universit\"at Duisburg-Essen 2005.
	\bibitem{mueller2009}
	{S. M\"uller, S. Heusler, A. Altland, P. Braun, and F. Haake,}
``Periodic-orbit theory of universal level correlations in quantum chaos'', 
	{\em New J.~Phys.}{\bf\,\,11}, 103025 (2009). 
	\bibitem{mueller2004}
	{S. M\"uller, S. Heusler, P. Braun, F. Haake, and A. Altland,}
	``Semiclassical foundation of universality in quantum chaos'', 
	{\em Phys.~Rev.~Lett.}{\bf\,\,93}, 014103 (2004).
	\bibitem{mueller2005}
	{S. M\"uller, S. Heusler, P. Braun, F. Haake, and A. Altland,}
	``Periodic-orbit theory of universality in quantum chaos'', 
	{\em Phys.~Rev.~E}{\bf\,\,72}, 046207 (2005).
	\bibitem{ratcliff}
	{J. Ratcliff, }{\em Foundations of Hyperbolic Manifolds},
	2nd edition, Springer, Berlin-Heidelberg-New York, 2006.
	\bibitem{Sieber1}
	{M. Sieber,} ``Semiclassical approach to spectral correlation functions'', 
	in {\em Hyperbolic Geometry and Applications in Quantum Chaos and Cosmology}, 
	Eds.~Bolte J. \& Steiner F., London Math.~Soc.~LNS 397, Cambridge University Press, 
	Cambridge-New York 2012, pp.~121-142.
	\bibitem{SieberRichter}
	{M. Sieber and K. Richter,} ``Correlations between periodic orbits and their r\^ole in spectral statistics'', {\em Physica Scripta}{\bf\,\,T90}, 128-133 (2001).
	\bibitem{Turek05}
	{M. Turek, D. Spehner, S. M\"uller, and K. Richter,}
	``Semiclassical form factor for spectral and matrix element fluctuation
	of multidimensional chaotic systems'',
	{\em Phys. Rev. E} {\bf 71}, 1-25 (2005).  
\end{thebibliography}
\end{document}